\documentclass[screen,sigconf]{acmart}

\makeatletter
\def\bbordermatrix#1{\begingroup \m@th
  \@tempdima 4.75\p@
  \setbox\z@\vbox{%
    \def\cr{\crcr\noalign{\kern2\p@\global\let\cr\endline}}%
    \ialign{$##$\hfil\kern2\p@\kern\@tempdima&\thinspace\hfil$##$\hfil
      &&\quad\hfil$##$\hfil\crcr
      \omit\strut\hfil\crcr\noalign{\kern-\baselineskip}%
      #1\crcr\omit\strut\cr}}%
  \setbox\tw@\vbox{\unvcopy\z@\global\setbox\@ne\lastbox}%
  \setbox\tw@\hbox{\unhbox\@ne\unskip\global\setbox\@ne\lastbox}%
  \setbox\tw@\hbox{$\kern\wd\@ne\kern-\@tempdima\left[\kern-\wd\@ne
    \global\setbox\@ne\vbox{\box\@ne\kern2\p@}%
    \vcenter{\kern-\ht\@ne\unvbox\z@\kern-\baselineskip}\,\right]$}%
  \null\;\vbox{\kern\ht\@ne\box\tw@}\endgroup}
\makeatother

\newcommand{\mycircled}[2][none]{%
 \tikz[baseline=(a.base)]\node[draw,circle,inner sep=0.5pt, outer sep=0pt,fill=#1](a){#2};
}
\newcommand{\myboxed}[2][none]{%
 \tikz[baseline=(a.base)]\node[draw,rectangle,inner sep=1pt, outer sep=0pt,fill=#1](a){#2};
}
\makeatletter
\newcommand\myparagraph{\def\@toclevel{4}%
  \@startsection{paragraph}{4}{\z@}%
  {-.2\baselineskip \@plus -2\p@ \@minus -.2\p@}%
  {-3.5\p@}%
  {\ACM@NRadjust{\@parfont}}}
\makeatother

\newcommand{\bigO}[1]{\mathchoice{O\left(#1\right)}{O(#1)}{O(#1)}{O(#1)}} % big O for complexity
\newcommand{\softO}[1]{\mathchoice{\tilde{O}\left(#1\right)}{\tilde O(#1)}{O\tilde{~}(#1)}{O\tilde{~}(#1)}} % soft O for complexity
\newcommand{\expmm}{\omega} % exponent for the cost of matrix multiplication
\newcommand{\Card}[1]{\# #1} % cardinality of a set
\newcommand{\proba}[1]{\mathbb{P}\left(#1\right)}
\newcommand{\probacst}{\rho} % constant for proba bounds
\newcommand{\ZZ}{\mathbb{Z}} % relative integers
\newcommand{\ZZp}{\mathbb{Z}_{> 0}} % positive integers
\newcommand{\field}{\mathbb{K}} % base field
\newcommand{\matspace}[2]{\field^{#1 \times #2}} % scalar matrix space
\newcommand{\genBy}[1]{\langle #1 \rangle} % ideal/module generated by #1
\newcommand{\module}{\mathcal{M}} % module for submatrix with zero left entries
\newcommand{\ideal}{\mathcal{I}} % ideal of multivariate polynomials
\newcommand{\xring}{\field[x]} % univariate polynomial ring in x
\newcommand{\xvecspace}[1]{\xring^{#1}} % univariate polynomial vectors in x
\newcommand{\xmatspace}[2]{\xring^{#1 \times #2}} % univariate polynomial matrix space in x
\newcommand{\xyring}{\field[x,y]} % bivariate polynomial ring in x,y
\newcommand{\trsp}[1]{#1^{\mathsf{T}}} % transpose
\newcommand{\itrsp}[1]{#1^{-\mathsf{T}}} % transpose
\newcommand{\diag}[1]{\operatorname{diag}(#1)}  % diagonal matrix with diagonal entries #1
\newcommand{\ident}[1]{\mathrm{I}_{#1}} % identity matrix of size #1 x #1

\newcommand{\rank}[1]{\operatorname{rank}(#1)} % matrix rank
\newcommand{\adj}[1]{\operatorname{adj}(#1)} % matrix rank

\newcommand{\perm}{\pi} % permutation matrix
\newcommand{\rspan}[1]{\operatorname{RowSpan}(#1)} % module spanned by rows of matrix
\newcommand{\cspan}[1]{\operatorname{ColSpan}(#1)} % module spanned by rows of matrix
\newcommand{\hyprow}{\mathcal{H}_{\mathrm{row}}} % generic row HNF property
\newcommand{\hypcol}{\mathcal{H}_{\mathrm{col}}} % generic column HNF property
\newcommand{\genl}{S} % generic column HNF property
\newcommand{\genr}{T} % generic column HNF property
\newcommand{\sddots}{\scalebox{.75}{$\ddots$}} % to scale ddots, e.g. in smallmatrix
\newcommand{\sumTuple}[1]{|#1|} % sum of entries in a tuple
\newcommand{\rdeg}[2][]{\mathrm{rdeg}_{{#1}}(#2)} % shifted row degree
\newcommand{\lmat}[2][]{\mathrm{lm}_{{#1}}(#2)} % shifted row degree
\newcommand{\dd}{D} % determinantal degree
\newcommand{\da}{\bar{D}} % degree of adjugate
\newcommand{\indset}{J} % set of indices
\makeatletter
\newcommand*{\rem}{%
  \nonscript\mskip-\medmuskip\mkern5mu%
  \mathbin{\operator@font rem}\penalty900\mkern5mu%
  \nonscript\mskip-\medmuskip
}
\makeatother
\newcommand{\lc}[1]{\operatorname{lc}(#1)} % leading coefficient
\newcommand{\lm}[1]{\operatorname{lm}(#1)}

\newcommand{\ydeg}[1]{\mathrm{deg}{}_{y}(#1)} % y-degree of bivariate poly
\newcommand{\xdeg}[1]{\mathrm{deg}{}_{x}(#1)} % x-degree of bivariate poly
\newcommand{\tdeg}[1]{\mathrm{tdeg}(#1)} % total degree of bivariate poly
\newcommand{\flagsfont}[1]{\textsf{#1}} % font for algo flags
\newcommand{\FlagSing}{\flagsfont{Singular}}
\newcommand{\FlagFail}{\flagsfont{Fail}}
\newcommand{\Output}{\flagsfont{out}}
\newcommand{\FlagCert}{\flagsfont{Cert}}
\newcommand{\FlagTrue}{\flagsfont{True}}

\newcommand{\FlagUnknown}{\flagsfont{Unknown}}
\newcommand{\sampleset}{\mathcal{S}}
\newcommand{\Glex}{\mathcal{G}_{\operatorname{lex}}}
\newcommand{\Gdrl}{\mathcal{G}_{\operatorname{drl}}}
\newcommand{\ord}{\preccurlyeq} % monomial order
\newcommand{\ordneq}{\prec} % monomial order (strict)
\newcommand{\lelex}{\mathrel{\ord_{\operatorname{lex}}}} % lexicographic order on the polynomial ring
\newcommand{\ltlex}{\mathrel{\ordneq_{\operatorname{lex}}}} % lexicographic order on the polynomial ring
\newcommand{\ledrl}{\mathrel{\ord_{\operatorname{drl}}}} % degree reverse lexicographic order on the polynomial ring
\newcommand{\ltdrl}{\mathrel{\ordneq_{\operatorname{drl}}}} % degree reverse lexicographic order on the polynomial ring
\usepackage{algorithm}
\usepackage{algpseudocodex}      %%% this is for the algorithms
\newcommand{\algoitem}{\textbullet\,} % item for algo input/output
\algrenewcommand\textproc{\textsf}   %%% was textsc
\algnewcommand{\algorithmicand}{\textbf{and} }
\algnewcommand{\algorithmicor}{\textbf{or} }
\algnewcommand{\OR}{\algorithmicor}
\algnewcommand{\AND}{\algorithmicand}
\algnewcommand{\InlineIf}[2]{% single line if-then
   \algorithmicif\ #1\ \algorithmicthen\ #2}
\algnewcommand{\InlineElse}[1]{% single line else
   \algorithmicelse\ #1}
\algnewcommand{\InlineFor}[2]{\algorithmicfor\ #1\ \algorithmicdo\ #2} % single line for loop
\algrenewcommand\Call[2]{\nameref{#1}\ifthenelse{\equal{#2}{}}{}{\ensuremath{(#2)}}}%
\newcommand{\algoName}[1]{Algorithm \nameref{#1}}
\makeatletter
\newcommand{\algoCaptionLabel}[2]{
     \caption[\textproc{#1}]{\textproc{#1}\ifthenelse{\equal{#2}{}}{}{$(#2)$} }%
     \NR@gettitle{\textproc{#1}}%
      \label{algo:#1}
     }%
\makeatother
\usetikzlibrary{patterns}
\usetikzlibrary{patterns.meta}

\usepackage[shortlabels]{enumitem}

\theoremstyle{acmdefinition}
\usepackage[capitalise]{cleveref}

\AddToHook{cmd/lemma/before}{\crefalias{theorem}{lemma}}
\AddToHook{cmd/proposition/before}{\crefalias{theorem}{proposition}}
\AddToHook{cmd/corollary/before}{\crefalias{theorem}{corollary}}
\AddToHook{cmd/appendix/before}{\crefalias{section}{appendix}}
\makeatletter
\AddToHook{env/algorithmic/begin}{\def\@currentcounter{ALG@line}}
\makeatother
\crefalias{ALG@line}{line}

\title%
[Computing submatrices of the HNF of a structured polynomial matrix]%
{Computing submatrices of the Hermite normal form \texorpdfstring{\\}{} of a structured polynomial matrix}

\author{J\'er\'emy Berthomieu}
\orcid{0000-0002-9011-2211}
\affiliation{
\institution{Sorbonne Université, CNRS, LIP6}
\city{F-75005 Paris}
  \country{France}
}

\author{Vincent Neiger}
\orcid{0000-0002-8311-9490}
\affiliation{
\institution{Sorbonne Université, CNRS, LIP6}
\city{F-75005 Paris}
  \country{France}
}

\author{Hugo Passe}
\orcid{0009-0005-4622-8531}
\affiliation{
\institution{Sorbonne Université, CNRS, LIP6}
\city{F-75005 Paris}
  \country{France}
}

\copyrightyear{2026}
\acmYear{2026}
\setcopyright{cc}
\setcctype{by}
\acmConference[ISSAC '26]{51st International Symposium on Symbolic and Algebraic Computation}{July 13--17, 2026}{Oldenburg, Germany}
\acmBooktitle{51st International Symposium on Symbolic and Algebraic Computation (ISSAC '26), July 13--17, 2026, Oldenburg, Germany}
\acmDOI{10.1145/3815436.3815457}
\acmISBN{979-8-4007-2595-1/2026/07}

\keywords{Hermite normal form; Structured polynomial matrix; Gr\"obner basis.}
\begin{document}

\begin{abstract}
Following several decades of successive algorithmic improvements, works from the 2010s have showed how to compute the Hermite normal form (HNF) of a univariate polynomial matrix within a complexity bound which is essentially that of polynomial matrix multiplication. Recently, several results on bivariate polynomials and Gr\"obner bases have highlighted the interest of computing determinants or HNFs of polynomial matrices that happen to be structured, with a small displacement rank. In such contexts, a small leading principal submatrix of the HNF often contains all the sought information. In this article, we show how the displacement structure can be exploited in order to accelerate the computation of such submatrices. To achieve this, we rely on structured linear algebra over the field thanks to evaluation-interpolation. This allows us to recover some rows of the inverse of the input matrix, from which we deduce the sought HNF submatrix via bases of relations.
\end{abstract}

\maketitle

\section{Introduction}
\label{sec:intro}

The echelonization of univariate polynomial matrices, through the Hermite
normal form (HNF), can be used to compute Gröbner bases of bivariate ideals.
Given polynomials \(f,g\) in \(\xyring\), under some mild assumptions, the HNF
of their Sylvester matrix \(M \in \xmatspace{n}{n}\) reveals the lexicographic
basis of \(\langle f,g\rangle\) \cite[Thm.\,4]{Lazard1985}. This
has recently been generalized to 
zero-dimensional ideals
\(\langle f_1,\ldots,f_{\ell}\rangle \subseteq \xyring\)
with \(\ell \ge 2\), by considering the HNF of a block-Toeplitz polynomial matrix
defined from the \(f_i\)'s \cite{SchostStPierre2023}. Exploiting fast
algorithms for polynomial matrices, this leads to improved complexity bounds,
specifically \(\softO{{\ell}^{\expmm}D^{\expmm}d}\) operations in \(\field\) if
\(d\) bounds the total degree of all \(f_i\)'s and \(\dd\) bounds the degree of
the ideal \cite[Sec.\,2]{SchostStPierre2023}; see also \cite{Dahan2024}. Here,
\(\expmm\) is an exponent for matrix multiplication, with \(2 < \expmm
\le 3\); one may take \(\expmm < 2.372\) \cite{VWXXZ24,AlmanEtAl2025}.

Sylvester and block-Toeplitz matrices are both structured matrices with small
displacement rank \cite{Pan2001}. In the above cost bound, this structure is
not exploited: the computation runs a general-purpose HNF algorithm, which uses
\(\softO{n^{\expmm} \deg(M)}\) operations in \(\field\)
\cite{LabahnNeigerZhou2017}.

Besides matrix structure, a particularity of the above situations is that one
is often satisfied with an \(m \times m\) leading principal submatrix of the
HNF of \(M\). Indeed, this submatrix contains the sought lexicographic basis as
soon as \(m\) exceeds the \(y\)-degree of all elements in this basis. One
usually has \(m \ll n\); in fact, for bases that have a generic form called
shape position, one has \(m=2\).

Even \(m=1\) is interesting, as the computed HNF entry provides a
minimal-degree univariate polynomial in \(\xring\) from the ideal. This is
closely related to the question of computing the determinant of \(M\), or its
largest Smith factor. If \(M\) is arbitrary, the best known cost bounds
for the latter tasks are also in \(\softO{n^{\expmm} \deg(M)}\)
\cite{ZhouLabahnStorjohann2015,LabahnNeigerZhou2017}. This might suggest
pessimism regarding the possibility of designing algorithms that
compute the \(m\times m\) submatrix of the HNF faster than by basically
finding the whole HNF and extracting the submatrix.

However, specifically for Sylvester matrices \(M \in \xmatspace{n}{n}\),
faster algorithms are known. The classical bivariate resultant algorithm
\cite[Sec.\,11.2]{GathenGerhard1999} finds the determinant of \(M\) in
\(\softO{n \dd}\) operations in \(\field\), where \(\dd\) is an a priori
bound on \(\deg(\det(M))\). This algorithm exploits the matrix structure thanks
to an evaluation-interpolation scheme that allows one to rely on fast
operations for structured matrices over \(\field\); Sylvester matrices have
displacement rank \(\alpha = 2\). Even for the most favorable degree patterns
in \(M\), using a general-purpose HNF algorithm would have a worse cost, in
\(\softO{n^{\expmm-1}\dd}\).

In this article, we generalize this approach. We design an algorithm which, for
any \(\alpha\) and \(m\), supports an input \(M\) with displacement rank \(\le
\alpha\) and computes the leading \(m \times m\) principal submatrix of the HNF
of \(M\). Its cost bound extends that of the bivariate resultant: these bounds
match as soon as \(\alpha\) and \(m\) are in \(\bigO{1}\). Our algorithm
follows an evaluation-interpolation scheme, and relies on fast, Las Vegas
randomized, structured linear algebra tools
\cite{BostanJeannerodMouilleron2017}. It includes optimizations under the next
assumptions:

\noindent\hfill
\begin{tabular}{c|c}
  $\hypcol$: generic \emph{column} HNF & $\hyprow$: generic \emph{row} HNF \\ \hline
  \scalebox{0.8}{%
    \begin{minipage}{0.54\columnwidth}
      \begin{equation}
        \label{eqn:generic_col_hnf}%
        MV =
        \begin{bmatrix}
            \bar{h}_0 & \bar{h}_1 & \cdots & \bar{h}_{n-1} \\
                      & 1  \\
                      &  & \ddots \\
                      & & & 1
        \end{bmatrix}
      \end{equation}
    \end{minipage}%
  }
  &
  \scalebox{0.8}{%
    \begin{minipage}{0.54\columnwidth}
      \begin{equation}
        \label{eqn:generic_row_hnf}%
        UM =
        \begin{bmatrix}
            h_0               \\
            h_1 & 1  \\
            \vdots &  & \ddots \\
            h_{n-1} & & & 1
        \end{bmatrix}
      \end{equation}
    \end{minipage}%
  }
\end{tabular}
\hfill { } \\
where $U, V \in \xmatspace{n}{n}$ are unimodular and
\(h_i, \bar{h}_j \in \xring\).
We refer to \cref{sec:structured_system:prelim,sec:hnf:prelim}
for some relevant definitions and notations.

\begin{theorem}
  \label{thm:hnf}%
  \algoName{algo:HermiteFormSubmatrix} takes as input displacement generators
  \(\genl,\genr \in \xmatspace{n}{\alpha}\) of degree at most \(d\) for a
  matrix \(M \in \xmatspace{n}{n}\), an integer \(m \in \{0,\ldots,n-1\}\), and
  bounds \(\dd \geq \deg(\det(M))\) and \(\da \geq \deg(\adj{M})\). Let
  $\Delta = \dd + \da + 1$. If \(M\) is singular, the algorithm returns
  \FlagFail{} or \FlagSing, using \(\softO{\alpha^{\expmm-1}n\Delta + \alpha n
  d}\) operations in \(\field\). If \(M\) is nonsingular, it uses
  \(
    \softO{\lceil m / \alpha \rceil\alpha^{\expmm-1}n\Delta + m^{\expmm-1}n\dd + \alpha n d}
  \)
  operations in \(\field\). It chooses elements at random from a finite
  \(\mathcal{S} \subseteq \field\) and returns \FlagFail{} or
  \FlagSing{} with probability at most \(1/2\) if $\Card{\sampleset} \geq 8
  \Delta \max(n (n+1) , 2\dd)$. If it does not return \FlagFail{} or
  \FlagSing, it returns the leading principal \(m \times m\) submatrix of the
  HNF of \(M\). Moreover, if the input satisfies \(\dd = \deg(\det(M))\),
  the cost reduces
  \begin{itemize}
    \item to \(\softO{\alpha^{\expmm-1}n\Delta + \alpha n d}\) if both \(\hypcol\)
      and \(\hyprow\) are true;
    \item to \(\softO{\alpha^{\expmm-1}n\Delta + m^{\expmm-1} \dd + \alpha n d}\)
      if only \(\hypcol\) is true;
    \item to \(\softO{\left\lceil m / \alpha \right\rceil\alpha^{\expmm-1}n\Delta
      + \alpha n d}\) if only \(\hyprow\) is true.
  \end{itemize}
\end{theorem}
%It chooses a finite subset \(\sampleset\) of \(\field\) of cardinality at least
%\(\Delta\), draws uniformly at random a subset of cardinality \(\Delta\) of
%\(\sampleset\), and picks at most \((2n-2)\Delta\) elements uniformly at random
%from \(\sampleset\)
%
Concerning the conditions in these items, note that the algorithm detects at
runtime whether \(\hypcol\) or \(\hyprow\) holds, in conjunction with \(\dd =
\deg(\det(M))\): none of these properties has to be known a priori. Note also
that frequent cases involve \(\dd\) and \(\da\) in \(\bigO{nd}\).

A natural perspective for further improvements comes from a recent breakthrough
concerning bivariate resultants
\cite{Villard2018,PernetSignargoutVillard2024}: under a genericity assumption,
the determinant of an \(n \times n\) Sylvester matrix of degree \(d\) can be
computed in \(\softO{n^{2 - 1/\expmm} d}\). In the shape position case, this
also gives access to the lexicographic basis \cite[Sec.\,7]{Villard2018}.
Whether the genericity assumption and algorithmic approach can be extended to
our context is beyond the scope of this paper.

We apply our algorithm to the change of monomial order.
A classical approach to find a lexicographic basis \(\Glex\) is to
first compute a Gröbner basis \(\Gdrl\) for a degree-refining order such as
\(\ledrl\), and then deduce \(\Glex\) from \(\Gdrl\) using a change of order
algorithm. We show how to construct, from \(\Gdrl\), a matrix \(M \in
\xmatspace{n}{n}\) of small displacement rank, with known degree of
determinant, and such that the HNF of \(M\) yields \(\Glex\). This leads to the
next result.

\begin{theorem}
  \label{thm:change_order}%
  Let \(\ideal \subseteq \xyring\) be a zero-dimensional ideal of degree
  \(\dd\) and \(\Gdrl\) be a minimal \(\ledrl\)-Gröbner basis of \(\ideal\)
  with \(x \ltdrl y\). We denote by $d_y$ the maximal \(y\)-degree of a
  polynomial in \(\Gdrl\) and $\ell$ the number of polynomials in \(\Gdrl\).
  \algoName{algo:HermiteFormSubmatrix} can be used to compute a reduced
  \(\lelex\)-Gröbner basis \(\Glex\) of \(\ideal\) with \(x \ltlex y\). It
  chooses elements at random from a finite \(\mathcal{S} \subseteq
  \field\) and returns \FlagFail{} or \FlagSing{} with probability at most
  \(1/2\) if $\Card{\sampleset} \geq 16 (\dd + d_y) \max(d_y (d_y+1) , 2\dd)$.
  It uses \(\softO{(\ell^{\expmm-1}  + m^{\expmm-1}) d_y \dd}\) operations in
  \(\field\), where \(m\) bounds the \(y\)-degree of all polynomials
  in \(\Glex\).
\end{theorem}

The parameters \(\ell\), \(d_y\), and \(\dd\) only depend on the set of
monomials involved in \(\Gdrl\). Since \(\ell \in
\bigO{\dd^{1/2}}\) and \(d_y \leq \dd\), for ideals that have the worst-case
monomials in \(\Gdrl\) the cost bound above becomes \(\softO{\dd^{(\expmm + 3)
/2}}\) as soon as \(m \in\bigO{\dd^{1/2}}\). Even in this unfavorable and
rather uncommon worst case, this improves upon the FGLM
algorithm \cite{FaugereGianniLazardMora1993}, which runs in \(\bigO{D^3}\).
Our result is also to be compared with a Gaussian elimination-based approach in
$\bigO{(d_{\max}-d_{\min})\dd^{\expmm}}$ \cite[Prop\,4.10]{Huot2013}, where
\(d_{\min}\) (resp.\ \(d_{\max}\)) is the minimal (resp.\ maximal) total degree
of the leading monomials of \(\Gdrl\). More precisely, the latter bound is for
computing so-called multiplication matrices, from which \(\Glex\) can
be found in \(\bigO{\dd^\expmm \log(\dd)}\) \cite[Ex.\,1.3,
Thm.\,1.7]{NeigerSchost2020}. There are wide ranges of parameters for which our
cost bound provides an improvement.
Finally, under assumptions, the change of order has been carried out in
\cite{BerthomieuNeigerSafey2022} by constructing a polynomial matrix
(unstructured, in that case) and computing its
HNF. Our result complements this with an
efficient HNF-based solution  in cases where the assumptions made in
\cite{BerthomieuNeigerSafey2022} do not hold. In particular, we do not
require that \(\Gdrl\) be reduced, and, most importantly, we remove
the so-called stability assumption.

\myparagraph{Outline.}

\Cref{sec:structured_system} combines an evaluation-interpolation scheme with
structured \(\field\)-linear algebra to solve modular polynomial matrix
equations. Based on this, \Cref{sec:hnf} describes the main algorithm
\nameref{algo:HermiteFormSubmatrix} and proves \cref{thm:hnf}. Finally,
\cref{sec:change_order} applies this algorithm to the bivariate change of
monomial order.

\section{Structured linear systems over \texorpdfstring{\(\xring/\genBy{A}\)}{K[x]/<A>}}
\label{sec:structured_system}

Our main algorithm in \cref{sec:hnf} relies on solving linear systems with
matrix \(M \in \xmatspace{n}{n}\), modulo some polynomial \(A \in \xring\).
Specifically, for some wide matrix \(X \in \xmatspace{m}{n}\) and tall matrix
\(Y \in \xmatspace{n}{m}\), it exploits the system solutions \(X M^{-1} \bmod
A\) and \(M^{-1} Y \bmod A\). In this section, we design a fast algorithm
for finding \(M^{-1} Y \bmod A\); this is easily modified to cover \(X M^{-1}
\bmod A\) (see \cref{app:left_rows_versions}). As highlighted above, our focus
is on matrices \(M\) with small displacement rank, and therefore, on obtaining
a complexity bound that takes this quantity into account.

We require that \(A\) splits into known distinct linear factors, that
is, \(A = \prod_{0 \le i < \Delta} (x - a_i)\) for known distinct
points \(a_0,\ldots,a_{\Delta-1} \in \field\). This is not a restriction in the
context of the later sections, where \(A\) can be chosen freely as long as its
degree is ``large enough''; note that this may require to work in a
small-degree extension of \(\field\) to guarantee that the field cardinality is
at least \(\Delta\). Using this special form of \(A\) allows us to follow an
evaluation-interpolation strategy. Once evaluated, \(M(a_i)\) is a matrix of
small displacement rank over the field \(\field\) and one can rely on existing
fast algorithms for structured matrices over fields
\cite{Pan2001,BostanJeannerodMouilleron2017}. Applying these
algorithms directly on \(M\), as a matrix over \(\field(x)\), would lead to
fractions of large degree and thus to bad performance.

\subsection{Preliminaries: structured matrices}
\label{sec:structured_system:prelim}

For univariate polynomials, we denote by \(\xring_{\dd}\) (resp.\
\(\xring_{<\dd}\)) the set of those of degree equal to \(\dd\) (resp.\ less
than \(\dd\)).

We use standard tools and terminology concerning structured matrices; the
reader may refer to \cite{Pan2001}. Define \(n \times n\) matrices
\[
  Z_0 =
  \left[\begin{smallmatrix}
    & & & 0 \\
    1 \\[-0.15cm]
    & \sddots \\
    & & 1
  \end{smallmatrix}\right]
  ,\quad
  Z_1 =
  \left[\begin{smallmatrix}
    & & & 1 \\
    1 \\[-0.15cm]
    & \sddots \\
    & & 1
  \end{smallmatrix}\right],
  \quad\text{and }
  L =
  \left[\begin{smallmatrix}
    & & & 1 \\[-0.15cm]
    & & \reflectbox{$\sddots$} \\
    & 1 \\
    1
  \end{smallmatrix}\right].
\]
For an \(n \times n\) matrix $M$ with
coefficients in \(\field\) (resp.\ \(\xring\)), its displacement rank is
the rank of \(Z_0 M - M \trsp{Z_1}\). Furthermore, for a given upper bound
\(\alpha\) on the displacement rank of \(M\), one calls displacement
generators for \(M\) any pair of \(n\times \alpha\) matrices \((\genl,\genr)\)
with coefficients in \(\field\) (resp.\ \(\xring\)) such that \(Z_0 M -
M\trsp{Z_1} = \genl \trsp{\genr}\). These generators are used as a compact
representation of \(M\), which is valid because the operator \(M \mapsto Z_0 M - M
\trsp{Z_1}\) is invertible: the data of \(\genl\) and \(\genr\) describes a unique
corresponding \(M\). This representation of \(M\) through \((\genl,\genr)\) is
compact in the sense that it only requires to store \(2\alpha n\) matrix
entries, instead of \(n^2\) for the dense representation of \(M\). In the case
of matrices over \(\xring\), there exist pairs of generators whose degrees do
not exceed that of \(M\), so that the compact representation uses \(2 \alpha n
(\deg(M) + 1)\) field elements.
%% NOTE sketch of proof:
% Indeed, one can show that the displaced
% polynomial matrix \(Z_0 M - M \trsp{Z_1}\) can be written as \(\genl
% \trsp{\genr}\) with \(\trsp{\genr}\) in so-called reduced form, and both
% \(\genl\) and \(\genr\) having degree at most \(\deg(M)\).
Proving this can be done through the notion of reduced forms (as in
\cref{sec:hnf:prelim} but for rectangular matrices); we will not detail
this here, both for conciseness and because the matrices we consider in our target
context have obvious small-degree generators (see \cref{sec:change_order}).

\subsection{Solving by evaluation-interpolation}
\label{sec:structured_system:solver}

One may have noted that our problem of computing \(M^{-1} Y \bmod A\) only
makes sense if \(M\) is invertible modulo \(A\). Equivalently, the determinant
of \(M\) should be coprime with \(A\). This property is not easily tested when
building \(A\), since one lacks information on \(M\) at this stage; for
example, \(\det(M)\) is not known yet. For this reason, we do not assume \(M\)
to be invertible modulo \(A\): our algorithm detects if that is not the
case and returns a flag ``\FlagSing'' accordingly.

One step of \cref{algo:StructuredRightSolve-WithPoints} computes the inverse of some
structured matrix over the base field, say \(M(a_i) \in \matspace{n}{n}\).
While this means that \(Z_0 M(a_i) - M(a_i) \trsp{Z_1}\) has small rank, the
inverse \(M(a_i)^{-1}\) has small displacement rank with respect to a
variant of the displacement operator from \cref{sec:structured_system:prelim}. For this
reason, the output of the inversion algorithm we call consists of generators
\((\bar{\genl},\bar{\genr})\) such that \(\trsp{Z_1} M(a_i)^{-1} - M(a_i)^{-1} Z_0 =
\bar{\genl} \trsp{\bar{\genr}}\) (note the permuted \(\trsp{Z_1}\) and \(Z_0\)).

\begin{algorithm}[ht]
	\algoCaptionLabel{StructuredRightSolve-WithPoints}{\genl,\genr,Y,(a_i)_{0 \le i < \Delta}}
  \begin{algorithmic}[1]

    \Require{%
        \algoitem displacement generators $\genl,\genr \in\xmatspace{n}{\alpha}$ that represent
        $M\in\xmatspace{n}{n}$ through \(Z_0 M - M \trsp{Z_1} = \genl \trsp{\genr}\) (see \cref{sec:structured_system:prelim}); \\
        \algoitem a matrix $Y\in\xmatspace{n}{m}$ of degree less than \(\Delta\); \\
        \algoitem pairwise distinct points $a_0,\ldots,a_{\Delta-1} \in \field$, for \(\Delta \in \ZZp\).%
    }

    \Ensure{Any one of \FlagFail;   % bad luck in the structured linear algebra preconditioning
      \((i,\FlagSing)\) for some \(0 \le i < \Delta\);     % not Failure and \(\det(M)(a_i) = 0\) and \(\det(M)(a_j) \neq 0\) for \(j < i\)
      and $F \in \xmatspace{n}{m}_{<\Delta}$ such that $M F = Y \bmod \prod_{0 \le i < \Delta} (x-a_i)$.}

	  \LComment{multipoint evaluation}
    \State $\genl^{(i)} \in \matspace{n}{\alpha} \gets \genl(a_i)$ for $0 \le i < \Delta$%
        \label{algo:StructuredRightSolve-WithPoints:evalG}%
    \State $\genr^{(i)} \in \matspace{n}{\alpha} \gets \genr(a_i)$ for $0 \le i < \Delta$%
        \label{algo:StructuredRightSolve-WithPoints:evalH}%
    \State $Y^{(i)} \in \matspace{n}{m} \gets Y(a_i)$ for $0 \le i < \Delta$%
      \label{algo:StructuredRightSolve-WithPoints:evalv}%

	  \LComment{at each point, structured linear algebra over \(\field\)}
    \State \(\sampleset \gets\) a finite subset of \(\field\)
	  \For{$i \in\{0,\ldots,\Delta-1\}$}
      \LComment{compute generators \((\bar{\genl},\bar{\genr})\) for the inverse of \(M(a_i)\):}
      \LComment{precisely, \(\trsp{Z_1} M(a_i)^{-1} - M(a_i)^{-1} Z_0 =
      \bar{\genl} \trsp{\bar{\genr}}\)}
      %% NOTE the tranposes are really like this, this is not a typo
      \State \(\flagsfont{gen\_inv} \gets \textproc{inv}(\genl^{(i)}, \genr^{(i)}, \sampleset)\)
      \Comment{Algorithm in \cite[Fig.\,7]{BostanJeannerodMouilleron2017}}
        \label{algo:StructuredRightSolve-WithPoints:inv}
      \State\InlineIf{\(\flagsfont{gen\_inv} = \FlagFail\)}{\Return \FlagFail}
      \State\InlineIf{\(\flagsfont{gen\_inv} = \FlagSing\)}{\Return \((i,\FlagSing)\)}
      \State write \(\flagsfont{gen\_inv} = (\bar{\genl},\bar{\genr})\), with \(\bar{\genl},\bar{\genr} \in \matspace{n}{\alpha}\)
      \LComment{multiply \(F^{(i)} \gets M(a_i)^{-1} Y(a_i)\)}
      \State \(F^{(i)} \in \matspace{n}{m} \gets\) product of the
      \((\trsp{Z_1},Z_0)\)-structured matrix \((\bar{\genl},\bar{\genr})\) by the dense
      matrix \(Y^{(i)}\)
      \Comment{algorithm of \cite[Sec.\,5]{BostanJeannerodMouilleron2017}}
        \label{algo:StructuredRightSolve-WithPoints:mul}
	  \EndFor

	  \LComment{interpolation}
    \State $F \in \xmatspace{n}{m} \gets$ the unique matrix of degree \(< \Delta\) such that \(F(a_i) = F^{(i)}\) for $0 \le i < \Delta$%
      \label{algo:StructuredRightSolve-WithPoints:interp}%

	\State \Return $F$
  \end{algorithmic}
\end{algorithm}

\begin{proposition}
  \label{prop:StructuredModularSystem}%
  \Cref{algo:StructuredRightSolve-WithPoints} uses \(\softO{\lceil m/\alpha \rceil
  \alpha^{\expmm-1} n \Delta + \alpha n d}\)
  % NOTE if wanting more details, logarithmic factors could be given
  operations in \(\field\), where \(d\) is an upper bound on both \(\deg(\genl)\)
  and \(\deg(\genr)\). It chooses a finite subset \(\sampleset\) of \(\field\) and
  picks at most \((2n-2)\Delta\) elements uniformly at random from \(\sampleset\).
  If \(\Card{\sampleset} \ge n(n+1)\), it returns \FlagFail{} with probability at most \(1 - (1 -
  n(n+1) / \Card{\sampleset})^\Delta\). If \(\Card{\sampleset} \ge \probacst n(n+1)\Delta\) for
  some \(\probacst > 1\), this bound is at most \(1/\probacst\).
  Assuming the algorithm does not return \FlagFail, it returns
  \((i,\FlagSing)\) if and only if \(\det(M)(a_i) = 0\) and \(\det(M)(a_j) \neq
  0\) for \(j < i\). If it returns neither \FlagFail{} nor \FlagSing, then
  \(M\) is invertible modulo \(A = \prod_{0 \le i < \Delta} (x-a_i)\) and it
  correctly returns the unique $F = M^{-1} Y \bmod A$ of degree \(< \Delta\).
\end{proposition}

\begin{proof}
  At each iteration, the call to \textproc{inv} at
  \cref{algo:StructuredRightSolve-WithPoints:inv} picks \(2n-2\) elements at random
  from \(\sampleset\), and this call succeeds (i.e., does not return \FlagFail) with
  probability at least \(1 - n(n+1) / \Card{\sampleset}\), according to
  \cite[Thm.\,6.6]{BostanJeannerodMouilleron2017} (and its proof, for the
  probability bound). Thus, over all iterations, at most \((2n-2)\Delta\) elements
  are picked at random from \(\sampleset\), and \cref{algo:StructuredRightSolve-WithPoints}
  returns something else than \FlagFail{} with probability at least \((1 -
  n(n+1) / \Card{\sampleset})^\Delta\).
  Assuming \(\Card{\sampleset} \ge \probacst n(n+1)\Delta\), this bound
  is at least \((1 - 1/(\probacst\Delta))^\Delta\), and the latter quantity grows as a function
  of \(\Delta\), with a value \(1 - 1/\probacst\) when \(\Delta=1\).

  If the \(i\)th call to \textproc{inv} yields \FlagSing, then the rank of the
  matrix represented by the generators \((\genl^{(i)},\genr^{(i)})\) is in
  \(\{0,\ldots,n-1\}\). That matrix is \(M(a_i)\), so in other words,
  \(\det(M)(a_i) = 0\). Conversely, if \(\det(M)(a_i) = 0\), the correctness of
  \textproc{inv} guarantees that \cref{algo:StructuredRightSolve-WithPoints} returns
  either \FlagFail{} or \FlagSing. As a result, when it returns neither of both
  flags, we have \(\det(M)(a_i) \neq 0\) for all \(i\), hence \(M\) is
  invertible modulo \(A\). In that case, by construction the output matrix
  \(F\) satisfies \(F(a_i) = M(a_i)^{-1} Y(a_i)\) and \(\deg(F) < \Delta\), hence
  \(F = M^{-1} Y \bmod A\).

  Having proved correctness, we turn to the cost bound. The evaluation and
  interpolation steps have quasi-linear complexity with respect to the size of
  the objects \cite[Chap.\,10]{GathenGerhard1999}. This means
  \(\softO{\strut (\alpha + m) n \Delta + \alpha n d}\)
  operations in \(\field\) for
  \cref{algo:StructuredRightSolve-WithPoints:evalG,algo:StructuredRightSolve-WithPoints:evalH,algo:StructuredRightSolve-WithPoints:evalv},
  and \(\softO{m n \Delta}\) operations in \(\field\) for
  \cref{algo:StructuredRightSolve-WithPoints:interp}. In a single iteration of the
  loop, the call to \textproc{inv} at \Cref{algo:StructuredRightSolve-WithPoints:inv}
  costs \(\softO{\alpha^{\expmm-1} n}\)
  \cite[Thm.\,6.6]{BostanJeannerodMouilleron2017}. If it succeeds and provides
  \((\trsp{Z_1}, Z_0)\)-generators \((\bar{\genl},\bar{\genr})\) for \(M(a_i)^{-1}\),
  then \cref{algo:StructuredRightSolve-WithPoints:mul} multiplies this
  \((\trsp{Z_1},Z_0)\)-structured matrix by the dense matrix \(Y^{(i)}\) using
  \(
    \softO{\min(\alpha,m)^{\expmm-2} \max(\alpha,m) n},
  \)
  operations, according to \cite[Thm.\,1.1]{BostanJeannerodMouilleron2017}.
  Thus, one iteration costs \(\softO{\lceil m/\alpha \rceil \alpha^{\expmm-1} n}\),
  and the whole loop uses \(\softO{\lceil m/\alpha \rceil \alpha^{\expmm-1} n \Delta}\)
  operations in \(\field\).
\end{proof}

\subsection{Build the modulus and solve}
\label{sec:structured_system:modulus}

Incorporating the choice of the points \(a_i\)'s, we obtain the modular
structured system solving procedure in \cref{algo:StructuredModularRightSolve}.

\begin{algorithm}[ht]
	\algoCaptionLabel{StructuredModularRightSolve}{\genl,\genr,Y,\Delta}
  \begin{algorithmic}[1]

    \Require{%
        \algoitem displacement generators $\genl,\genr \in\xmatspace{n}{\alpha}$ that represent
        $M\in\xmatspace{n}{n}$ through \(Z_0 M - M \trsp{Z_1} = \genl \trsp{\genr}\) (see \cref{sec:structured_system:prelim}); \\
        \algoitem a matrix $Y\in\xmatspace{n}{m}$ of degree less than \(\Delta\); \\
        \algoitem an integer \(\Delta \in \ZZp\).%
    }

    \Ensure{Any one of \FlagFail;
      $((a_j)_{0\le j \le i}, \FlagSing)$
      with \(0 \le i < \Delta\) and $a_j \in \field$;
      and \((A,F)\) where \(A \in \xring_{\Delta}\)
      is coprime with \(\det(M)\) and $F \in \xmatspace{n}{m}_{<\Delta}$
      is such that $MF = Y \bmod A$.}

    \State\InlineIf{\(\field\) has cardinality \(< \Delta\)}{\Return \FlagFail}

    \State $\sampleset \gets$ a finite subset of $\field$ of cardinality \(\ge \Delta\)
      \label{algo:StructuredModularRightSolve:finite_subset}

    \State $\{a_0,\ldots,a_{\Delta-1}\} \gets$ subset of \(\sampleset\), sampled uniformly
    at random among all subsets of \(\sampleset\) of cardinality \(\Delta\)
      \label{algo:StructuredModularRightSolve:choose_points}

    \State $F \gets \Call{algo:StructuredRightSolve-WithPoints}{\genl,\genr,Y,(a_i)_{0 \le i < \Delta}}$
      \label{algo:StructuredModularRightSolve:solve}
    \State\InlineIf{$F = \FlagFail$}{\Return \FlagFail}
      \label{algo:StructuredModularRightSolve:fail}
    \State\InlineIf{$F = (i,\FlagSing)$}{\Return $((a_j)_{0\le j \le i}, \FlagSing)$}
      \label{algo:StructuredModularRightSolve:singular}
    \State\Return \((\prod_{0 \le i < \Delta} (x-a_i), F)\)
      \label{algo:StructuredModularRightSolve:success}
  \end{algorithmic}
\end{algorithm}

\begin{proposition}
  \label{prop:StructuredModularSolve}%
  Assume \(\field\) has cardinality \(\ge \Delta\).
  \Cref{algo:StructuredModularRightSolve} uses \(\softO{\lceil m/\alpha
    \rceil \alpha^{\expmm-1} n \Delta + \alpha n d}\)
  % NOTE if wanting more details, logarithmic factors could be given
  operations in \(\field\), where \(d\) is an upper bound on both \(\deg(\genl)\)
  and \(\deg(\genr)\). It chooses a finite subset \(\sampleset\) of \(\field\) of
  cardinality at least \(\Delta\), draws uniformly at random a subset of
  cardinality \(\Delta\) of \(\sampleset\), and picks at most \((2n-2)\Delta\) elements
  uniformly at random from \(\sampleset\). It always returns \FlagFail{} or \FlagSing{}
  if \(\det(M) = 0\). Otherwise, defining \(\dd = \deg(\det(M)) \ge 0\),
  if for some \(\probacst > 1\) one has
  \[
    \Card{\sampleset} \ge \max(\probacst n(n+1) \Delta, \probacst\dd \Delta + \min(\dd,\Delta) - 1),
  \]
  then the algorithm returns \FlagFail{} or \FlagSing{} with probability
  \(\le 2/\probacst\).

  If it returns $((a_j)_{0\le j \le i}, \FlagSing)$, the $a_j$'s are
  pairwise distinct points in \(\field\) such that \(\det(M)(a_i) = 0\)
  and \(\det(M)(a_j) \neq 0\) for \(j < i\). If it returns \((A,F)\), then \(A \in
  \xring\) has degree \(\Delta\) and is coprime with \(\det(M)\), and $F \in
  \xmatspace{n}{m}$ is the unique matrix of degree \(<
  \Delta\) such that $F = M^{-1} Y \bmod A$.
\end{proposition}

\begin{proof}
  From the properties of \cref{algo:StructuredRightSolve-WithPoints} stated in
  \cref{prop:StructuredModularSystem}, one directly obtains:
  the complexity bound and correctness along with the claimed properties of the
  output;
  the fact that \cref{algo:StructuredModularRightSolve:solve} always
  returns \FlagFail{} or $(0,\FlagSing{})$ when \(\det(M) = 0\);
  and the upper bound \(1/\probacst\) on the probability that the algorithm
  returns \FlagFail, under the stated assumption on \(\Card{\sampleset}\).
  Thus, it is enough to show that it returns \FlagSing{} with
  probability \(\le 1/\probacst\), when \(\det(M) \neq 0\).

  If the call at \cref{algo:StructuredModularRightSolve:solve} yields $(i,
  \FlagSing)$, then \(\det(M)(a_i) = 0\). In particular, the probability that \cref{algo:StructuredModularRightSolve}
  returns \FlagSing{} is less than or equal to the probability that
  the sampled set \(\{a_0,\ldots,a_{\Delta-1}\}\) intersects with the set of
  roots of \(\det(M)\). Let \(\Gamma \subseteq \field\) be this set of roots.
  Our assumption \(\det(M) \neq 0\) ensures that \(\Card{\Gamma} \le D\),
  hence
  \[
    \proba{\{a_0,\ldots,a_{\Delta-1}\} \cap \Gamma = \emptyset}
    = \frac{\binom{\Card{(\sampleset \setminus \Gamma)}}{\Delta}}{\binom{\Card{\sampleset}}{\Delta}}
    \geq \frac{\binom{\Card{\sampleset} - \dd}{\Delta}}{\binom{\Card{\sampleset}}{\Delta}}.
  \]
  By basic manipulations, the right-hand side can be rewritten as
  \[
    \frac{\binom{\Card{\sampleset}-\dd}{\Delta}}{\binom{\Card{\sampleset}}{\Delta}}
    = \prod_{i=0}^{\dd-1} \frac{\Card{\sampleset} - \Delta - i}{\Card{\sampleset} - i}
    = \prod_{i=0}^{\Delta-1} \frac{\Card{\sampleset} - \dd - i}{\Card{\sampleset} - i} .
  \]
  All terms involved in these products are positive, since our assumption
  on \(\Card{\sampleset}\) implies \(\Card{\sampleset} \ge \dd + \Delta\). This leads to the lower bound
  \[
    \frac{\binom{\Card{\sampleset}-\dd}{\Delta}}{\binom{\Card{\sampleset}}{\Delta}}
    \ge
    \max\left\{
      \left(1 - \frac{\dd}{\Card{\sampleset} + 1 - \Delta}\right)^{\Delta}
      ,
      \left( 1 - \frac{\Delta}{\Card{\sampleset} + 1 - \dd} \right)^{\dd}
    \right\}
    .
  \]
  Then, the assumption
  \(\Card{\sampleset} \ge \probacst \dd \Delta + \min(\dd,\Delta) - 1\) implies that at
  least one of the terms in the above \(\max\) is greater than or equal to \(1
  - 1/\probacst\). Indeed, say for example \(\Card{\sampleset} \ge \probacst\dd
  \Delta + \Delta - 1\), then one has
  \[
    \left( 1 - \frac{\dd}{\Card{\sampleset} + 1 - \Delta} \right)^{\Delta}
    \ge
    \left( 1 - \frac{1}{\probacst\Delta} \right)^{\Delta}
    \ge
    1 - \frac{1}{\probacst}.
  \]
  As a result, \(\proba{\{a_0,\ldots,a_{\Delta-1}\} \cap \Gamma \neq
  \emptyset} \le 1/\probacst\).
\end{proof}

\section{Computing submatrices of the HNF}
\label{sec:hnf}

In this section, we present our main algorithm. Given a nonsingular matrix $M
\in\xmatspace{n}{n}$, it computes a canonical basis \(B \in \xmatspace{m}{m}\)
of some module defined from \(M\) (see \cref{eqn:module}). We identify several conditions
under which the HNF of \(B\) coincides with a submatrix of the HNF of \(M\); in
particular, this is the case 
when one targets 
for 
the leading principal submatrix.
The algorithm starts by computing parts of 
\(M^{-1}\), and then reconstructs \(B\) via a relation basis. Some optimizations are described
in favorable cases, such as when \(M\) satisfies some kind of genericity
property (see \cref{sec:hnf:problem}), and when 
\(\deg(\det(M))\)
is known. The latter occurs for example when \(M\) is in shifted
reduced form (see \cref{sec:change_order} for a concrete such case).

\subsection{Preliminaries: polynomial matrices}
\label{sec:hnf:prelim}

For an \(m\times n\) matrix \(M\), we write \(M_{ij}\) for its entry \((i,j)\),
with indexing \(0 \le i < m\) and \(0 \le j < n\). We use standard submatrix
notation such as \(M_{*j}\) for the column \(j\), and \(M_{*J}\)
for the submatrix formed by columns \(j \in J \subseteq \{0,\ldots,n-1\}\).
For a polynomial matrix $M \in\xmatspace{n}{n}$ and a degree shift \(s =
(s_j)_j \in \ZZ^n\), the \(s\)-row degree of \(M\) is the tuple
denoted by $\rdeg[s]{M} = (t_i)_i \in (\ZZ \cup \{-\infty\})^n$,
whose \(i\)th entry is $t_i = \max_{0\le j < n} (\deg(M_{ij}) + s_j)$.
The $s$-leading matrix of $M$ is the constant matrix $\lmat[s]{M} \in
\matspace{n}{n}$ whose entry $(i,j)$ is the coefficient of degree $t_i - s_j$
of $M_{ij}$ (which we take as $0$ if $t_i = -\infty$).

The matrix $M$ is said to be in \(s\)-reduced form if \(\lmat[s]{M}\) is
invertible; in $s$-weak Popov form if $\lmat[s]{M}$ is invertible and lower
triangular; and in \(s\)-Popov form if it is in \(s\)-weak Popov form and
in addition \(\lmat[(0,\ldots,0)]{\trsp{M}} = \ident{n}\)
\cite{Popov1972,Kailath1980,BeckermannLabahnVillard1999}. The latter condition
means, equivalently, that $\deg(M_{ij}) < \deg(M_{ii})$ for all \(i\neq j\) and
$M_{ii}$ is monic. Similarly, \(M\) is said to be in Hermite normal form
(shortened as HNF) if \(H\) is lower triangular with \(\lmat[0]{\trsp{M}} =
\ident{n}\) \cite{Hermite1851,Kailath1980}. One can check that a matrix \(M\) in HNF
is also in \(s\)-Popov form for the shift \(s = (0,\dd,2\dd,\ldots,(n-1)\dd)\),
where \(\dd = \deg(\det(M))\).

We call \(\xring\)-row span of \(M\), denoted by \(\rspan{M}\), the
\(\xring\)-module generated by its rows, that is,  \(\{u M \mid u \in
\xmatspace{1}{n}\}\).
Let \(\module\) be a \(\xring\)-submodule of \(\xmatspace{1}{n}\). Then
\(\module\) is free \cite[Chap.\,12]{DummitFoote2004}, and if it has rank
\(n\), then there exists \(M \in\xmatspace{n}{n}\) which is nonsingular
(meaning \(\det(M)\neq0\)) and such that \(\module = \rspan{M}\). In this
situation, we say that \(M\) is a basis of \(\module\). There is a unique such
basis \(H \in \xmatspace{n}{n}\) in HNF, and for a given \(s \in
\ZZ^n\), there is a unique basis \(P \in \xmatspace{n}{n}\) of \(\module\) in
\(s\)-Popov form \cite{BeckermannLabahnVillard1999}.

Bases of \(\module\) are related via left-unimodular equivalence: \(P =
UM\) for some matrix \(U\) which is unimodular, that is, \(\det(U) \in
\field\setminus\{0\}\). We say that \(H\) is the HNF of \(M\), and \(P\) is the
\(s\)-Popov form of \(M\). In this paper we also consider the \emph{column} HNF
of \(M\), which is a similar notion corresponding the column span of \(M\): one
may simply consider that it is the transpose of the HNF of \(\trsp{M}\).

In our HNF algorithms, we will make use of relation bases. The relation module
of $\mu \in\xring\setminus\{0\}$ and $F \in\xmatspace{m}{n}$ is
defined as
\[
  \mathcal{R}(\mu,F) = \left\{ p \in \xmatspace{1}{m} \mid pF = 0 \bmod \mu\right\}.
\]
This is a free module of rank $m$, with bases represented as nonsingular
matrices in $\xmatspace{m}{m}$. Its \(s\)-Popov bases, and in
particular its HNF basis, can be computed efficiently (see \cite{Neiger2016} and \cref{app:relbas}):

\begin{proposition}
  \label{prop:HNFRelBas}%
  There is an algorithm $\textproc{HNFRelBas}(\mu, F)$ which, given $\mu
  \in\xring_{\dd}$ and $F \in\xmatspace{m}{n}_{<\dd}$ for some \(\dd
  \in \ZZp\), returns the HNF basis of $\mathcal{R}(\mu,F)$ using
  \(\softO{m^{\expmm-1}n\dd}\) operations in \(\field\).
\end{proposition}

\subsection{Problem and notation}
\label{sec:hnf:problem}

Fix \(n \in \ZZp\), and consider a parameter \(m \in \{1,\ldots,n\}\).
Hereafter, we call \(m\)-index tuple any tuple of indices \(\indset =
(j_0,\ldots,j_{m-1})\) such that \(0 \le j_0 < \cdots < j_{m-1} < n\).
Our goal is to find the canonical basis of the module formed by vectors in the
row span of \(M\) whose entries with index not in \(\indset\) are zero.
Denoting the complement set \(\{0,\ldots,n-1\} \setminus \indset\) by \(\{k_0,
\ldots, k_{n-m-1}\}\), where \(k_0 < \cdots < k_{n-m-1}\), we define \(\perm\)
as the \(n \times n\) permutation matrix such that
\[
  [j_0 \; \cdots \; j_{m-1} \; k_0 \; \cdots \; k_{n-m-1}] \perm = [0 \; 1 \; \cdots \; n-1],
\]
that is, the only nonzero entries of \(\perm\) are \(1\)'s at indices
\((i,j_i)\) for \(0
\le i < m\) and \((i,k_i)\) for \(m \le i < n\). The \(\xring\)-module we are
interested in is
\begin{equation}
  \label{eqn:module}%
  \module_{\indset} = \{p \in \xmatspace{1}{m} \mid [p \;\; 0] \perm \in \rspan{M} \}.
\end{equation}
Since \(\det(M) \xmatspace{1}{m} \subseteq \module_{\indset} \subseteq
\xmatspace{1}{m}\), this is a free module of rank \(m\) \cite[Chap.\,12,
Thm.\,4]{DummitFoote2004}.

In particular, when \(\indset = (0,\ldots,m-1)\), \(\perm\) is the identity
and \(\module_\indset\) is the row span of the leading principal \(m \times m\)
submatrix of the HNF of \(M\), as we prove in \cref{lem:hnf_submat}.
Although it may be possible to reduce to this case through row and column
permutations of \(M\), we prefer to avoid this approach and keep a general
\(\indset\), since these permutations would typically increase the displacement rank
of \(M\).

Our algorithm is more efficient and has better properties for matrices \(M\)
whose column or row span has some generic property. These are exactly the
assumptions \(\hypcol\) and \(\hyprow\) of
\cref{eqn:generic_col_hnf,eqn:generic_row_hnf}.
In that case, \(h_0\) is a
constant multiple of \(\det(M)\), and the first row of \(M^{-1}\) is
\(U_{0*}/h_0\); similar remarks hold for \(\bar{h}_0\) and \(V_{*0}\).

We start by presenting an ingredient common to both cases in
\cref{sec:hnf:inverse}, which consists in computing such subsets of rows or
columns of \(M^{-1}\) thanks to the algorithm of \cref{sec:structured_system}.
Then, in \cref{sec:hnf:cols,sec:hnf:rows}, we state the main properties that
underlie the correctness of the algorithm, which is described in
\cref{sec:hnf:algo}.

\subsection{Finding rows or columns of the inverse of \texorpdfstring{\(M\)}{M}}
\label{sec:hnf:inverse}

Using \cref{sec:structured_system} to find $F = M^{-1} \trsp{[1 \; 0 \;
\cdots \; 0]} \bmod A$ allows us to recover the first column of \(M^{-1}\), if
\(\deg(A)\) is large enough. More generally, \(m\) columns or
\(m\) rows of \(M^{-1}\) can be computed, without impacting the cost as long as
\(m \in \bigO{\alpha}\). The algorithm requires some a priori bound on
\(\deg(\det(M))\) and \(\deg(\adj{M})\): general ones are \(\dd = n \deg(M)\)
and \(\da = (n-1) \deg(M)\), but better bounds may be available, for example if
\(M\) has unbalanced row or column degrees, or if it is in shifted reduced or
shifted Popov form (\cref{sec:hnf:prelim}).

\begin{algorithm}[ht]
	\algoCaptionLabel{InverseCols}{\genl,\genr,\indset,\dd,\da}
  \begin{algorithmic}[1]

    \Require{%
        \algoitem displacement generators $\genl,\genr \in\xmatspace{n}{\alpha}$ that represent
        $M\in\xmatspace{n}{n}$ through \(Z_0 M - M \trsp{Z_1} = \genl
        \trsp{\genr}\) (see \cref{sec:structured_system:prelim}); \\
        \algoitem \(\indset = (j_0,\ldots,j_{m-1})\) with \(0 \le j_0 < j_1 < \cdots < j_{m-1} < n\); \\
        \algoitem degree bounds \(\dd \ge \deg(\det(M))\) and \(\da \ge \deg(\adj{M})\).%
    }

    \Ensure{Any one of \FlagFail; \FlagSing;
      and \((\mu, \mu N_{*\indset})\) where \(N = M^{-1}\) and
      \(\mu\) is the monic least common denominator of \(N_{*\indset}\).}

    \State \(\Delta \in \ZZp \gets \dd + \da + 1\)
    \label{algo:InverseOfColumns:params}

    % \State\InlineIf{\(\Card{\field} < \Delta\)}{\Return \FlagFail}
    \State \(Y \in \matspace{n}{m} \gets\) submatrix of columns \((j_0,\ldots,j_{m-1})\) of \(\ident{n}\)
    \label{algo:InverseOfColumns:Y}

    \State \(\Output \gets \Call{algo:StructuredModularRightSolve}{\genl, \genr, Y, \Delta}\)
    \label{algo:InverseOfColumns:solve}

    \State\InlineIf{\(\Output = \FlagFail\)}{\Return \FlagFail}
    \label{algo:InverseOfColumns:fail}

    \State\InlineIf{\(\Output = (*, \FlagSing)\)}{\Return \FlagSing}
    \label{algo:InverseOfColumns:sing}

    \State \((A,F) \gets \Output\), with \(A \in \xring_{\Delta}\) and \(F \in \xmatspace{n}{m}_{<\Delta}\)
    \label{algo:InverseOfColumns:AF}

    \For{\(i \in \{0,\ldots,n-1\}\)} \Comment{rational reconstruction of \(N_{ij_k}\)}
      \For{\(k \in \{0,\ldots,m-1\}\)}
        \State \((f_{ik},g_{ik}) \gets\) nonzero polynomials with
        \(\deg(g_{ik}) \le \dd\), \(\deg(f_{ik}) \le \da\), and \(g_{ik} F_{ik} = f_{ik} \bmod A\)
        \label{algo:InverseOfColumns:ratrecon}
      \EndFor
    \EndFor
    \State \(g \in \xring_{\le \dd} \gets\) least common multiple of \((g_{ik})_{0 \le i < n, 0 \le k < m}\)
    \label{algo:InverseOfColumns:lcm}

    \State\Return \((g , [f_{ik}g/g_{ik}]_{0 \le i < n, 0 \le k < m})\)
  \end{algorithmic}
\end{algorithm}

\begin{proposition}
  \label{prop:columns_of_inverse}%
  \Cref{algo:InverseCols} uses \(\softO{\lceil m/\alpha
  \rceil\alpha^{\expmm-1} n \Delta + \alpha n d}\) operations in \(\field\),
  where \(d\) is an upper bound on both \(\deg(\genl)\) and \(\deg(\genr)\).
  If it does not return \FlagFail{} nor \FlagSing, it correctly computes
  $(\mu,\mu N_{*\indset})$.
\end{proposition}

\begin{proof}
  By \cref{prop:StructuredModularSolve}, \cref{algo:InverseOfColumns:solve}
  uses \(\softO{\lceil m/\alpha \rceil \alpha^{\expmm-1} n \Delta + \alpha n
  d}\) operations, and if \cref{algo:InverseOfColumns:AF} is
  reached, \(A\) is coprime with \(\det(M) \neq 0\) and \(F = M^{-1} Y \bmod A\). By
  choice of \(Y\), this means \(F = N_{*\indset} \bmod A\).

  The rational reconstruction at \cref{algo:InverseOfColumns:ratrecon} costs
  \(\softO{\Delta}\) \cite{BrentGustavsonYun1980}, and has a solution since the
  sum of degree constraints is precisely \(\Delta-1 = \deg(A) - 1\)
  \cite[Sec\,5.7]{GathenGerhard1999}. For all reconstructions, this is
  a total of \(\softO{m n \Delta}\) operations in \(\field\),
  which is bounded by \(\softO{\lceil m/\alpha \rceil \alpha^{\expmm-1} n
  \Delta}\).

  Combining the
  equation \(g_{ik} N_{ij_k} = g_{ik} F_{ik} =
  f_{ik}\bmod A\) with
  \(\det(M) N_{ij_k} = \adj{M}_{ij_k}\), we
  get \(f_{ik} \det(M) = g_{ik} \adj{M}_{ij_k} \bmod A\). By construction,
  both sides have degree less than \(\Delta = \deg(A)\), hence \(f_{ik} \det(M)
  = g_{ik} \adj{M}_{ij_k}\). Since \(\det(M)M^{-1}\) is a polynomial matrix, \(\mu\) divides \(\det(M)\),
  and 
  multiplying by $\mu /
  \det(M)$ gives \(\mu f_{ik} = g_{ik} \mu N_{ij_k}\). The latter is a
  polynomial identity since \(\mu N_{ij_k} \in \xring\).
  In particular
  \(g_{ik}\) divides \(\mu\), since \(\gcd(g_{ik},f_{ik}) = 1\).
  This proves that the least common multiple \(g\) at
  \cref{algo:InverseOfColumns:lcm} divides \(\mu\); besides, its computation
  costs \(\softO{n \dd}\). Now we prove that \(\mu\) divides \(g\). Following
  \(\mu f_{ik} = g_{ik}\mu N_{ij_k}\) we have \([f_{ik}g/g_{ik}]_{i,k} = g
  N_{*\indset}\). The left-hand side has entries in \(\xring\), hence \(\mu\)
  divides \(g\), and \(g = \mu\) since both are monic. This also shows that the
  returned matrix \([f_{ik}g/g_{ik}]_{i,k}\) is \(\mu N_{*\indset}\). Computing this
  matrix from \(f_{ik}\)'s, \(g_{ik}\)'s, and \(g\) costs \(\softO{m n
  \Delta}\) operations, which is in \(\softO{\lceil m/\alpha \rceil
  \alpha^{\expmm-1} n \Delta}\).
\end{proof}

\subsection{Case of generic column HNF}
\label{sec:hnf:cols}

The above algorithm yields an efficient verification that \(M\) has the
generic column HNF property, when \(\deg(\det(M))\) is known:

\begin{lemma}
  \label{lem:verify_hypcol}%
  Let \(M\in \xmatspace{n}{n}\) be nonsingular and \(N = M^{-1}\). If
  \(\mu \in \xring\setminus\{0\}\) is the least common denominator of
  \(N_{*0}\), then \(M\) satisfies \(\hypcol\) if and only if
  \(\deg(\mu) = \deg(\det(M))\). If \(\mu \in \xring\setminus\{0\}\) is the
  least common denominator of \(N_{0*}\), then \(M\) satisfies \(\hyprow\) if
  and only if \(\deg(\mu) = \deg(\det(M))\).
\end{lemma}

\begin{proof}
  Let \(V \in \xmatspace{n}{n}\) be the unimodular matrix such that \(MV\) is
  in column HNF. Then, the first diagonal entry $h \in
  \xring\setminus\{0\}$ of \(MV\) is such that $M V_{*0} = \trsp{[h
  \; 0 \; \cdots \; 0]}$. Left-multiplying by \(N\) gives $V_{*0} = h
  N_{*0}$,  hence \(h\) is a multiple of \(\mu\). From the identity
  $\frac{h}{\mu} V_{*0} = \mu N_{*0}$, it follows that \(\mu/h\) divides all
  entries of \(V_{*0}\); yet $\gcd(V_{*0}) = 1$ since \(\deg(\det(V)) = 0\).
  Thus \(\mu = h\).
  The stated equivalence follows from the fact that
  \(M\) satisfies \(\hypcol\) if and only if \(h = \det(M)/\lc{\det(M)}\).
  The statement with \(\hyprow\) is proved in the same way, by considering
  \(U_{0*} M = [h \; 0 \; \cdots \; 0]\) where \(U \in \xmatspace{n}{n}\) is
  the unimodular matrix such that \(UM\) is in (row) HNF.
\end{proof}

Having computed \((\mu, \mu N_{*0})\), and assuming we have been able to check
that \(M\) satisfies \(\hypcol\), a canonical basis of \(\module_\indset\) can
be obtained as a relation basis for the subvector \(\mu N_{\indset0}\),
modulo \(\mu\):

\begin{lemma}
  \label{lem:exploit_hypcol}%
  Let \(M\in \xmatspace{n}{n}\) satisfy the property \(\hypcol\). Let \(N =
  M^{-1}\) and \(\mu = \det(M)\). Then, one has \(\rspan{M} =
  \mathcal{R}(\mu, \mu N_{*0})\) and \(\module_\indset =
  \mathcal{R}(\mu, \mu N_{\indset0})\) for any \(m\)-index tuple \(\indset\).
\end{lemma}

\begin{proof}
  By assumption there is a unimodular matrix \(V\) such that \(MV\) is upper
  triangular with diagonal \((\mu,1,\ldots,1)\), and thus \(UMV =
  \diag{\mu,1,\ldots,1}\) for some unimodular \(U\). Note that \(V_{*0} = \mu
  N_{*0} = \adj{M}_{*0}\). Then, for a row vector \(p \in \xmatspace{1}{n}\), one has
  \begin{align*}
    p \in \rspan{M} & \Leftrightarrow p V = u U M V \text{ for some } u \in \xmatspace{1}{n} \\
                    & \Leftrightarrow p \mu N_{*0} = u_0 \mu \text{ for some } u_0 \in \xring.
  \end{align*}
  This proves \(\rspan{M} = \mathcal{R}(\mu, \mu N_{*0})\).
  Then, by definition of \(\module_\indset\), for any
  \(p \in \xmatspace{1}{m}\) one has
  \[
    p \in \module_\indset
    \Leftrightarrow [p \;\; 0] \perm \in \rspan{M}
    \Leftrightarrow [p \;\; 0] \perm \mu N_{*0} = 0 \bmod \mu .
  \]
  The latter is equivalent to \(p \in \mathcal{R}(\mu, \mu N_{\indset0})\), since
  the definition of \(\perm\) along with the zeroes in \([p \;\; 0]\) imply
  that \([p \;\; 0] \perm N_{*0} = p N_{\indset0}\).
\end{proof}

\subsection{Using rows of the inverse}
\label{sec:hnf:rows}

When \(M\) does not satisfy \(\hypcol\), we make use of the whole submatrix
\(N_{\indset*}\) of \(M^{-1}\), instead of only \(N_{\indset0}\) in the
favorable case of \cref{sec:hnf:cols}. Our approach has two ingredients,
presented in the next statements. First, we express \(\module_\indset\) as a
module of relations involving \(\mu\) and \(\mu N_{\indset*}\). Second, we
describe several cases where the HNF basis of \(\module_\indset\) is equal to
the submatrix \(H_{\indset\indset}\) of the HNF of \(M\).

\begin{lemma}
  \label{lem:invrows_module}%
  Let \(M \in \xmatspace{n}{n}\) be nonsingular, let \(\indset = (j_k)_{0\le k
  < m}\) be an \(m\)-index tuple, and let \(\mu \in
  \xring\setminus\{0\}\) be such that \(\mu N_{\indset*}\)
  has entries in \(\xring\). Then, \(\module_\indset =
  \mathcal{R}(\mu,\mu N_{\indset*})\).
\end{lemma}
\begin{proof}
  Let \(p \in \xmatspace{1}{m}\). By definition, \(p\) is in \(\module_{\indset}\) if
  and only if there exists \(u \in \xmatspace{1}{n}\) such that \([p \;\; 0] \perm = uM\),
  or equivalently, \([p \;\; 0] \perm \mu N = u\mu\). By construction of
  \(\perm\) and due to the zeroes in \([p \;\; 0]\), one has \([p \;\; 0] \perm
  \mu N = p \mu N_{\indset*}\). Therefore \(p\) is in \(\module_{\indset}\)
  if and only if \(p \mu N_{\indset*} = 0 \bmod \mu\), which means
  \(p \in \mathcal{R}(\mu,\mu N_{\indset*})\).
\end{proof}

\begin{lemma}
  \label{lem:hnf_submat}%
  Let \(M \in \xmatspace{n}{n}\) be nonsingular and \(H \in
  \xmatspace{n}{n}\) be its HNF. Let \(\indset = (j_k)_{0\le k < m}\) be an
  \(m\)-index tuple and let \(B \in \xmatspace{m}{m}\) be the HNF
  basis of \(\module_{\indset}\). Consider the assertions:
  \begin{enumerate}[({\ref*{lem:hnf_submat}.}1)]   % ,leftmargin=0.2cm,itemindent=0.7cm]
    \item\label{lem:hnf_submat:ones}
      \(H_{jj} = 1\) for all \(j \in \{0,1,\ldots,j_{m-1}\} \setminus \indset\);
    \item\label{lem:hnf_submat:leading}
      \(\indset = (0,\ldots,m-1)\);
    \item\label{lem:hnf_submat:fills_space}
      one has \(\deg(B_{00}) + \cdots + \deg(B_{\bar{m}\bar{m}}) =
      \deg(\det(M))\) for some \(\bar{m} \ge 0\) such that
      \(\{0,\ldots,\bar{m}\} \subseteq \indset\).
  \end{enumerate}
  If any of
  \ref{lem:hnf_submat:ones},
  \ref{lem:hnf_submat:leading},
  and \ref{lem:hnf_submat:fills_space}
  holds, then \(H_{\indset\indset} = B\).
\end{lemma}
\begin{proof}
  Observe that \(H_{\indset\indset}\) is in HNF. Thus, for all three
  cases, all that needs to be proved is that the row span of
  \(H_{\indset\indset}\) is \(\module_{\indset}\).

  \emph{Case \ref{lem:hnf_submat:ones}.} Let \(p \in \module_{\indset}\). Then
  \([p \;\; 0] \perm\) is in \(\rspan{M}\), hence \([p \;\; 0] \perm = u H\) for some \(u \in
  \xmatspace{1}{n}\). It is enough to prove \(u_j = 0\) for all \(j \not \in
  \indset\): indeed, this gives \([p \;\; 0] \perm = u_\indset H_{\indset*}\),
  and restricting to the columns in \(\indset\) yields the sought conclusion
  \(p = u_\indset H_{\indset\indset}\). First, since the last nonzero entry of
  \([p \;\; 0] \perm\) is at index \(j_{m-1}\) and \(H\) is lower triangular
  with nonzero diagonal entries, we deduce \(u_j = 0\) for all \(j > j_{m-1}\).
  Second, for \(j \in \{0,1,\ldots,j_{m-1}\} \setminus \indset\), since by
  assumption the \(j\)th column of \(H\) is the \(j\)th column of the identity
  matrix, the fact that the \(j\)th entry of \([p \;\; 0] \perm\) is zero
  implies \(u_j=0\).

  The above paragraph proves \(\rspan{H_{\indset\indset}} \supseteq
  \module_{\indset}\) and we now prove the other inclusion, by taking \(p =
  H_{j_k\indset}\) for \(0 \le k < m\) and showing \(p \in
  \module_{\indset}\). For this, it is enough to observe that \(H_{j_kj} = 0\)
  for all \(j \not \in \indset\): indeed, this gives \([p \;\; 0] \perm =
  H_{j_k*}\), hence \([p \;\; 0] \perm\) is in the row span of \(M\).
  This observation \(H_{j_kj} = 0\) follows from the triangularity of \(H\) if
  \(j > j_k\), and from the fact that the \(j\)th column of \(H\) is the
  \(j\)th column of the identity matrix otherwise.

  \emph{Case \ref{lem:hnf_submat:leading}} is a direct consequence since
  \(\{0,1,\ldots,j_{m-1}\} \setminus \indset = \emptyset\).

  \emph{Case \ref{lem:hnf_submat:fills_space}.} Let $I = \{0,\ldots,\bar{m}\}$.
  Then $H_{II}$ is the HNF basis of $\module_I$ thanks to the case
  \ref{lem:hnf_submat:leading}. Let $p \in \module_I$ and $\bar{p} = [p \;\; 0]
  \in \xmatspace{1}{m}$. Then $\bar{p} \in \module_{\indset}$, and thus
  $\bar{p} = \bar{v}B$ for some \(\bar{v} \in \xmatspace{1}{m}\). Since $B$ is
  nonsingular and lower triangular, $\bar{v}_i=0$ for all $i > \bar{m}$, hence
  $p = \bar{v}_I B_{II}$. We have proved $\module_I \subseteq
  \rspan{B_{II}}$, which implies that \(\det(B_{II})\) divides
  $\det(H_{II})$.
  Besides, $\det(H_{II})$ divides $\det(M)$ since $M$ is lower
  triangular. The equality of degrees in \ref{lem:hnf_submat:fills_space} then
  ensures $\det(B_{II}) = \det(H_{II}) = \det(M)/\lc{\det(M)}$. Thus $H_{jj}
  = 1$ holds for all $j \in\{0,\ldots,n-1\}$ that are not in $I$, and in
  particular for those not in $\indset$: using the case \ref{lem:hnf_submat:ones}
  concludes the proof.
\end{proof}

The assertion \ref{lem:hnf_submat:ones} is not easy to test without a priori
knowledge on \(H\). On the other hand, as soon as \(\deg(\det(M))\) is known,
one may easily build an algorithm upon the assertion
\ref{lem:hnf_submat:fills_space}.

\subsection{Main algorithm and proof of Theorem~\ref{thm:hnf}}
\label{sec:hnf:algo}

\begin{algorithm}[ht]
	\algoCaptionLabel{HermiteFormSubmatrix}{\genl,\genr,\indset,\dd,\da}
  \begin{algorithmic}[1]

    \Require{%
        \algoitem displacement generators $\genl,\genr \in\xmatspace{n}{\alpha}$ that represent
        $M\in\xmatspace{n}{n}$ through \(Z_0 M - M \trsp{Z_1} = \genl \trsp{\genr}\) (see \cref{sec:structured_system:prelim}); \\
        \algoitem \(\indset = (j_0,\ldots,j_{m-1})\) with \(0 = j_0 < j_1 < \cdots < j_{m-1} < n\); \\
        \algoitem degree bounds \(\dd \ge \deg(\det(M))\) and \(\da \ge \deg(\adj{M})\).%
    }

    \Ensure{Any one of \FlagFail; \FlagSing; and \((B,\FlagCert)\) where
      \(B \in \xmatspace{m}{m}\) is the HNF basis of \(\module_{\indset}\),
      and the flag \(\FlagCert\) is in \(\{\FlagTrue, \FlagUnknown\}\) with
    \(\FlagCert=\FlagTrue \Rightarrow B = H_{\indset\indset}\) where \(H\) is
  the HNF of \(M\).}

    \LComment{use column \(N_{*0}\) to detect case \(\hypcol\) (generic column HNF)}

    \State \(\Output \gets \Call{algo:InverseCols}{\genl,\genr,(0),\dd,\da}\)
    \label{algo:HermiteFormSubmatrix:inv_cols}
    \State\InlineIf{\(\Output = \FlagFail\) \OR \(\Output = \FlagSing\)}{\Return \Output}
    \label{algo:HermiteFormSubmatrix:cols_failure}

    \State \((\mu, c) \in \xring \times \xmatspace{n}{1} \gets \Output\)
    \label{algo:HermiteFormSubmatrix:mu_c}

    \If{\(\deg(\mu) = \dd\) \AND \(\gcd(\mu,c_0) = 1\)} \Comment{\(\hypcol\) and \(\hyprow\)}
    \label{algo:HermiteFormSubmatrix:if_hypcol_hyprow}
      \State \(B \in \xmatspace{m}{m} \gets \diag{\mu,1,\ldots,1}\)
      \State\InlineFor{\(1 \le i < m\)}{\(B_{i0} \gets - c_{j_i} / c_{0} \bmod \mu\)}
        \label{algo:HermiteFormSubmatrix:hrow_hcol_relbas}
      \State \Return \(B\), \FlagTrue
        \label{algo:HermiteFormSubmatrix:return_hrow_hcol}
    \EndIf
    \If{\(\deg(\mu) = \dd\)} \Comment{\(\hypcol\)}
    \label{algo:HermiteFormSubmatrix:if_hypcol}
      \State \(B \in \xmatspace{m}{m} \gets \textproc{HNFRelBas}(\mu, c_\indset)\)
        \label{algo:HermiteFormSubmatrix:hcol_relbas}
      \State\InlineIf{\(\indset = \{0,\ldots,m-1\}\)}{\Return \(B\), \FlagTrue}
          \label{algo:HermiteFormSubmatrix:hcol:test_leading}

      \State\InlineIf{\(\deg(B_{00}) + \cdots + \deg(B_{\bar{m}\bar{m}}) = \dd\) for some
            \(\bar{m} \ge 0\) such that \(\{0,\ldots,\bar{m}\} \subseteq \indset\)}{%
          \Return \(B\), \FlagTrue}
          \label{algo:HermiteFormSubmatrix:hcol:test_fills_space}
      \State \Return \(B\), \FlagUnknown
    \EndIf

    \LComment{\(\hypcol\) not satisfied: use rows \(N_{\indset*}\)}

    \State \(\Output \gets
    \Call{algo:InverseRows}{\genl,\genr,\indset,\dd,\da}\) \Comment{defined in
    \cref{app:left_rows_versions}}
    \label{algo:HermiteFormSubmatrix:inv_rows}

    \State\InlineIf{\(\Output = \FlagFail\) \OR \(\Output = \FlagSing\)}{\Return \Output}
    \label{algo:HermiteFormSubmatrix:rows_failure}

    \State \((\mu, R) \in \xring \times \xmatspace{m}{n} \gets \Output\)
    \label{algo:HermiteFormSubmatrix:mu_r}

    \If{\(\deg(\mu) = \dd\) and \(\gcd(\mu,R_{00},\ldots,R_{0,n-1}) = 1\)}\Comment{\(\hyprow\)}
      \label{algo:HermiteFormSubmatrix:gcd}%
      \State \(B \in \xmatspace{m}{m} \gets \diag{\mu,1,\ldots,1}\)
      \For{\(1 \le i < m\)}
        \State \([\begin{smallmatrix} \mu & 0 \\ b & 1\end{smallmatrix}]
            \gets \textproc{HNFRelBas}\left(\mu, [\begin{smallmatrix} R_{0*} \\ R_{i*} \end{smallmatrix}]\right)\);
            \;\;
            \(B_{i0} \gets b\)
        \label{algo:HermiteFormSubmatrix:hrow_relbas}
      \EndFor
      \State \Return \(B\), \FlagTrue
        \label{algo:HermiteFormSubmatrix:return_hrow_rel_bas}
    \EndIf

    \State \(B \in \xmatspace{m}{m} \gets \textproc{HNFRelBas}(\mu, R)\)
        \label{algo:HermiteFormSubmatrix:general_relbas}

    \State\InlineIf{\(\indset = \{0,\ldots,m-1\}\)}{\Return \(B\), \FlagTrue}
        \label{algo:HermiteFormSubmatrix:test_leading}

    \State\InlineIf{\(\deg(B_{00}) + \cdots + \deg(B_{\bar{m}\bar{m}}) = \dd\) for some
          \(\bar{m} \ge 0\) such that \(\{0,\ldots,\bar{m}\} \subseteq \indset\)}{%
        \Return \(B\), \FlagTrue}
        \label{algo:HermiteFormSubmatrix:test_fills_space}

    \State \Return \(B\), \FlagUnknown
  \end{algorithmic}
\end{algorithm}

Based on \cref{algo:HermiteFormSubmatrix}, we now prove \cref{thm:hnf}.

\myparagraph{Probability.} \Cref{algo:HermiteFormSubmatrix} returns
\(\FlagFail\) or \(\FlagSing\) if \cref{algo:HermiteFormSubmatrix:inv_cols} or
\cref{algo:HermiteFormSubmatrix:inv_rows} returns \(\FlagFail\) or \(\FlagSing\).
The probability that they do so is given by
\cref{prop:StructuredModularSolve}. If \(M\) is nonsingular, for \(\rho >
1\) and a choice of \(\sampleset\) such as in
\cref{prop:StructuredModularSolve}, this probability is bounded by
\(\frac{2}{\rho}\). Hence the probability of \cref{algo:HermiteFormSubmatrix}
returning \(\FlagFail\) or \(\FlagSing\) is upper bounded by \(\frac{4}{\rho}\).
Taking \(\rho = 8\) and using that \(\Delta = \dd + \da + 1 > \dd\) yields the statement in
\cref{thm:hnf}.

\myparagraph{Correctness.} Assume that neither
\cref{algo:HermiteFormSubmatrix:inv_cols} nor
\cref{algo:HermiteFormSubmatrix:inv_rows} returns \(\FlagFail\) or
\(\FlagSing\).
% Case H_col and H_row: show hypcol
Then, at \cref{algo:HermiteFormSubmatrix:mu_c},
\cref{prop:columns_of_inverse} gives that \(\mu\) is the least common
denominator of \(N_{*0}\), and \(c = \mu N_{*0}\).
Assume \cref{algo:HermiteFormSubmatrix:if_hypcol_hyprow} is evaluated to true,
\cref{lem:verify_hypcol} guarantees that \(\hypcol\) is true and \(\mu =
\det(M)\).
% Case H_col and H_row: show hyprow
Moreover, \cref{lem:exploit_hypcol} gives that \(\rspan{M} =
\mathcal{R}(\det(M),c)\).  Then for any row \(p \in \rspan{M}\) there exists \(u
\in \xvecspace{n}\) such that \(pc = \det(M) u\). If \(p = [p_0 \; 0]\) with
\(\lambda \in \xring\), one has \(p_0 c_0 = \det(M) u_0\). Therefore, as
\(\gcd(\det(M),c_0) = 1\), then \(\det(M) \vert p_0\), thus \(\hyprow\) is true
by \cref{lem:verify_hypcol}.
% Correctness in case H_col H_row
Following \cref{lem:hrow_relbas_onecol}, \(B\) at
\cref{algo:HermiteFormSubmatrix:return_hrow_hcol} is a basis of
\(\mathcal{R}(\det(M),c_{\indset})\), that is \(\module_{\indset}\) by
\cref{lem:exploit_hypcol}. Moreover, \(\deg(B_{00}) = \deg(\det(M))\) and
\cref{lem:hnf_submat:fills_space} gives that \(B = H_{\indset\indset}\).

% Case H_col only
Otherwise, if \cref{algo:HermiteFormSubmatrix:if_hypcol} is evaluated to true,
\cref{lem:verify_hypcol} still guarantees that \(\hypcol\) is true.
% Correctness in case H_col only
\Cref{lem:exploit_hypcol} shows that \(\module_{\indset} =
\mathcal{R}(\mu,c_\indset)\). Furthermore, if the test
\cref{algo:HermiteFormSubmatrix:hcol:test_leading} (resp.
\cref{algo:HermiteFormSubmatrix:hcol:test_fills_space}) is true, then by
\cref{lem:hnf_submat:leading} (resp. \cref{lem:hnf_submat:fills_space}), \(B\)
is indeed \(H_{\indset\indset}\).

%% Case not H_col
Now, if \cref{algo:HermiteFormSubmatrix}, reaches
\cref{algo:HermiteFormSubmatrix:inv_rows}, still assuming it does not return
\FlagFail{} nor \FlagSing, \cref{algo:HermiteFormSubmatrix:mu_r} is reached and
the specification of \algoName{algo:InverseRows} gives that \(\mu\) is the least
common denominator of \(N_{\indset*}\), and \(R = \mu N_{\indset*}\).
% Case H_row
If \cref{algo:HermiteFormSubmatrix:gcd} is evaluated to true, then following
\cref{lem:verify_hypcol}, \(\hyprow\) is true and \(\mu = \det(M)\).
% Correctness of case H_row
By \cref{lem:invrows_module}, \(\module_{\indset} = \mathcal{R}(\det(M),R)\).
Let \(p = [p_0 \; 0] \in \mathcal{R}(\det(M),R)\) with \(p_0 \in \xring\). Then
\(p_0 R_{0*} = 0 \bmod \det(M)\), and since \(\gcd(\det(M),R_{0*}) = 1\),
\(\det(M)\) divides \(p_0\). Hence, the HNF basis of \(\module_{\indset}\)
satisfies \cref{lem:hnf_submat:fills_space} and is equal to \(H_{\indset
\indset}\). Now since \(H\) satisfies \(\hyprow\), so does the HNF basis of
\(\mathcal{R}(\det(M),R)\) and \cref{lem:hrow_relbas_manycol} gives that \(B\)
at \cref{algo:HermiteFormSubmatrix:return_hrow_rel_bas} is \(H_{\indset \indset}\).
% Case none of H_col H_row
Finally, if \cref{algo:HermiteFormSubmatrix} reaches
\cref{algo:HermiteFormSubmatrix:general_relbas} then by
\cref{lem:invrows_module}, \(B\) at
\cref{algo:HermiteFormSubmatrix:general_relbas} is a basis of
\(\module_{\indset}\). Moreover, if the test
\cref{algo:HermiteFormSubmatrix:test_leading} (resp.
\cref{algo:HermiteFormSubmatrix:test_fills_space}) is true, then by
\cref{lem:hnf_submat:leading} (resp. \cref{lem:hnf_submat:fills_space}), \(B\)
is indeed \(H_{\indset\indset}\).

%%% Complexity bound
\myparagraph{Complexity bound.} We use \(\Delta = \dd + \da + 1\).
\Cref{algo:HermiteFormSubmatrix:inv_cols} costs \(\softO{\alpha^{\expmm-1} n
\Delta + \alpha n d}\) by \cref{prop:StructuredModularSolve}. If
\cref{algo:HermiteFormSubmatrix:if_hypcol_hyprow} is true then both \(\hypcol\)
and \(\hyprow\) are true. In that case the cost of
\cref{algo:HermiteFormSubmatrix:hrow_hcol_relbas} is \(\softO{m \dd}\) (\(m\)
products of polynmoials of degree at most \(\dd\) and one modular inversion as in
\cite[Th\,11.10]{GathenGerhard1999}) which is bounded by the previous term.
Thus, if both \(\hypcol\) and \(\hyprow\) are true,
\cref{algo:HermiteFormSubmatrix} runs in \(\softO{\alpha^{\expmm-1} n
\Delta + \alpha n d}\).

If \cref{algo:HermiteFormSubmatrix} reaches
\cref{algo:HermiteFormSubmatrix:if_hypcol} and evaluates it to true then
\cref{algo:HermiteFormSubmatrix:hcol_relbas} costs \(\softO{m^{\expmm-1} \dd}\)
by \cref{prop:HNFRelBas}. The total cost of \cref{algo:HermiteFormSubmatrix}
is then \(\softO{\alpha^{\expmm-1} n \Delta + \alpha n d + m^{\expmm-1} \dd}\).

Now if \cref{algo:HermiteFormSubmatrix:inv_rows} is reached, it costs
\(\softO{\lceil m/\alpha \rceil\alpha^{\expmm-1} n \Delta + \alpha n d}\) by
\cref{prop:StructuredModularSolve}. The test at
\cref{algo:HermiteFormSubmatrix:gcd} costs \(\softO{n\dd}\) (\(n\) gcd's as in
\cite[Cor\,11.9]{GathenGerhard1999}) which is bounded by the cost of
\cref{algo:HermiteFormSubmatrix:inv_cols}.  If these tests are evaluated to
true, then \cref{algo:HermiteFormSubmatrix:hrow_relbas} is called \(m\)
times. Following \cref{prop:HNFRelBas}, this costs \(\softO{m n \dd}\). This gives a total
cost of \(\softO{\lceil m/\alpha \rceil\alpha^{\expmm-1} n \Delta + \alpha n
d}\) if only \(\hyprow\) is tested to true.

Finally, if none of the tests are true,
\cref{algo:HermiteFormSubmatrix:general_relbas} computes a relation basis in
\(\softO{m^{\expmm-1} n \dd}\), according to \cref{prop:HNFRelBas}.
In that case, the total cost of \cref{algo:HermiteFormSubmatrix} is
\(\softO{\lceil m/\alpha \rceil\alpha^{\expmm-1} n \Delta + m^{\expmm-1} n \dd +
\alpha n d}\).

\section{Change of monomial order for bivariate ideals}
\label{sec:change_order}

In this section, we apply the above results to efficiently compute
the lexicographic basis of a zero-dimensional ideal in \(\xyring\), given by a
\(\ledrl\)-Gr\"obner basis. In particular, we prove \cref{thm:change_order}.
The construction of \(M\) in \cref{sec:change_order:buildM} is
illustrated in \cref{app:biv_matrix_example}.

\subsection{Preliminaries: bivariate polynomials}
\label{sec:change_order:prelim}

For basic notions and definitions, we refer to \cite{CoxLittleOShea2010}. We
denote by \(\ledrl\) the degree reverse lexicographic order and by \(\lelex\)
the lexicographic order,  with $x \ltdrl y$ and $x\ltlex y$. For \(f \in
\xyring\), its \(\ledrl\)-leading monomial, \(x\)-degree, \(y\)-degree, and
total degree are written as \(\lm{f}\), \(\xdeg{p}\), \(\ydeg{p}\), and
\(\tdeg{p}\). For a given \(n\in\ZZp\), we define the \(\xring\)-module
\(\xyring_{<(\cdot,n)} = \{f \in \xyring \mid \ydeg{f} < n\}\).

An ideal \(\ideal \subseteq \xyring\) is said to be zero-dimensional if the
quotient \(\xyring/\ideal\) has finite dimension as a \(\field\)-vector space.
In this case, \(\dd = \dim_\field(\xyring/\ideal)\) is called the degree of
\(\ideal\). If \((g_0,\ldots,g_{\ell-1})\) is a \(\ledrl\)-Gr\"obner basis of
\(\ideal\) \cite[Ch.\,2]{CoxLittleOShea2010}, then the set of monomials that
are not divisible by any of the \(\lm{g_i}\)'s forms a vector space basis of
\(\xyring/\ideal\); thus there are \(\dd\) such monomials \cite[Ch.\,5,
\S3]{CoxLittleOShea2010}. The ideal \(\ideal\) is said to be in  shape position
if its \(\lelex\)-Gr\"obner basis has the form \((f_0(x), y - f_1(x))\) for
some \(f_0 \in \xring_{\dd}\) and \(f_1 \in \xring_{<\dd}\).

Having a matrix \(M \in \xmatspace{n}{n}\) in shifted reduced form, or in one
of the stronger above forms, yields useful information on the degree of its
determinant and of its adjugate \(\adj{M} \in \xmatspace{n}{n}\). We write
\(\sumTuple{s} = s_0 + \cdots + s_{n-1}\) for the sum of a tuple of integers.

\begin{lemma}
  \label{lem:degdet_degadj}%
  Let $s \in\ZZ^n$, let $M\in\xmatspace{n}{n}$ be in \(s\)-reduced form, and
  let $\dd = \sumTuple{\rdeg[s]{M}} - \sumTuple{s}$. The determinant of
  \(M\) has degree \(\dd\) and its leading coefficient is \(\lc{\det(M)} =
  \det(\lmat[s]{M}) \in \field\setminus\{0\}\). Furthermore, $\deg(\adj{M})
  \leq \dd + \max(s) - \min(s)$.
  %% NOTE regarding the adjugate, we also have
  % \(\cdeg[-s]{\adj{M}} = (\dd-t_0, \ldots, \dd-t_n)\)
  % \(\rdeg[t]{\adj{M}} = (\dd+s_0,\ldots,\dd+s_{n-1})\)
  % where \(t = \rdeg[s]{M} \in \ZZ^n\).
  % One is proved below, the other is a consequence after also considering
  % the fact that \(\adj{M}\) is \(-s\)-column reduced;
  % see for example \cite[Lem.\,3.1]{RosenkildeStorjohann2021}. But do
  % we need these stronger properties?
\end{lemma}

%\begin{lemma}[{\cite[Thm.\,6.3-13]{Kailath1980}}]
%  \label{lem:pdp}%
%  Let $M\in\xmatspace{m}{n}$ with $m \leq n$ be an $s$-reduced matrix for $s
%  \in\ZZp^n$. Then for every $v\in\xvecspace{m}$,
%  \[
%    \rdeg[s]{vM} = \max_{0\leq i < m} \left\{\deg(v_i) + \rdeg[s]{M_{i*}} \right\} = \rdeg[d]{v},
%  \]
%  where $d = \rdeg[s]{M}$.
%\end{lemma}

\begin{proof}
  The properties of \(\det(M)\) are proved in \cite[p.\,384]{Kailath1980}.
  Combining the identity \(\adj{M} M = \det(M) \ident{n}\) with the predictable
  degree property,
  see \cite[Thm.\,6.3-13]{Kailath1980}
  and \cite[Lem.\,3.6]{BeckermannLabahnVillard1999},
  yields \(\rdeg[t]{\adj{M}} = (\dd+s_0,\ldots,\dd+s_{n-1})\), where \(t =
  \rdeg[s]{M}\). Now, since the shift \(\bar{t} = (t_0 - \min(t), \ldots,
  t_{n-1} - \min(t))\) has nonnegative entries, we have \(\deg(\adj{M}) \le
  \max(\rdeg[\bar{t}]{\adj{M}})\). Since \(M\) has no zero row, we deduce
  \(\min(t) \ge \min(s)\), and therefore
  \begin{align*}
    & \max(\rdeg[\bar{t}]{\adj{M}}) = \max(\rdeg[t]{\adj{M}}) - \min(t) \\
    & \;\; \le \max(\rdeg[t]{\adj{M}}) - \min(s) = \dd + \max(s) - \min(s).
    \qedhere
  \end{align*}
\end{proof}

\subsection{Building the matrix \texorpdfstring{\(M\)}{M} from a Gr\"obner basis}
\label{sec:change_order:buildM}

Let $\ideal \subseteq \xyring$ be a zero-dimensional ideal of degree \(\dd\)
and $\mathcal{G} = (g_0,\ldots,g_{\ell-1})$ be a minimal \(\ledrl\)-Gröbner
basis of \(\ideal\) (hence with \(\ell \ge 2\)), sorted by increasing
$y$-degree of the leading monomial. Our goal is to compute $\Glex$, the
reduced \(\lelex\)-Gröbner basis of $\ideal$.

\begin{lemma}
  \label{lem:appli:gbprops}
  One has \(\ydeg{\lm{g_i}} < \ydeg{\lm{g_{i+1}}}\) for all \(i\) in \(\{0,\ldots,\ell-2\}\);
  $\lm{g_0}$ is a power of $x$; $\lm{g_{\ell-1}}$ is a power of $y$;
  and $\ydeg{g_{\ell-1}} = \ydeg{\lm{g_{\ell-1}}}$.
\end{lemma}

\begin{proof}
  This follows from the minimality and chosen sorting of $\mathcal{G}$, from
  the zero-dimensionality of $\ideal$ \cite[Chap.\,5, \S3,
  Thm.\,6]{CoxLittleOShea2010}, and from the fact that \(g_{\ell-1}\) and
  \(\lm{g_{\ell-1}}\) have the same total degree.
\end{proof}

Define $n_i = \ydeg{\lm{g_{i+1}}} - \ydeg{\lm{g_i}} > 0$ for $0 \le i < \ell-1$ and
$n_{\ell-1} = \max(1, \max_{0\leq i\leq \ell-2} (n_i+\ydeg{g_i}) - \ydeg{g_{\ell-1}})$.
%% NOTE without the max, we still have \(n_{\ell-1} \ge 0\), indeed:
%% n_{\ell-1} >= n_{\ell-2} + \ydeg{g_{\ell-2}} - \ydeg{g_{\ell-1}}
%%            >= n_{\ell-2} + \ydeg{\lm{g_{\ell-2}}} - \ydeg{g_{\ell-1}}
%%             = n_{\ell-2} + \ydeg{\lm{g_{\ell-2}}} - \ydeg{\lm{g_{\ell-1}}}
%%             = 0
%% For simplicity we put the \max(1,..) which overshoots by at most 1,
%% in most cases we would be fine without the max, but by putting it
%% this avoids to discuss corner cases where e.g. GB_drl == GB_lex in which
%% we do not recover from M because we discarded g_{\ell-1}
By \cref{lem:appli:gbprops}, one has $\ydeg{\lm{g_i}} = n_0 + \cdots
+ n_{i-1}$ for $1 \le i < \ell$.
% follows from the telescopic sum
% $\sum_{i=0}^{k-1} n_i = \ydeg \lm{g_{k}} - \ydeg \lm{g_0}$
% which is $\ydeg \lm{g_{k}}$ by \cref{lem:appli:gbprops}.
Consider the tuple of polynomials of length $n = n_0 + \cdots + n_{\ell-1}$,
\[
  \bar{\mathcal{G}} =
  (g_0, yg_0, \ldots, y^{n_0-1}g_0,
  \ldots,
  g_{\ell-1}, \ldots, y^{n_{\ell-1}-1}g_{\ell-1}),
\]
and the matrix $M \in\xmatspace{n}{n}$ whose \(r\)th row represents
the \(r\)th polynomial of $\bar{\mathcal{G}}$, viewed in $\xring[y]$, on the
basis $(1,y,\ldots,y^{n-1})$. Thus, \(M\) is formed by
\(\ell\) stacked Toeplitz blocks: for $0 \le i < \ell$, $0 \le k < n_i$, and $0
\le j \le n-n_i$,
the entry \(M_{n_0 + \cdots + n_{i-1}+k,j+k}\) is the coefficient of $y^j$ in the polynomial \(g_i\), viewed in
\(\xring[y]\). For the latter
bound on \(j\), using \cref{lem:appli:gbprops} we get
\[
  \ydeg{g_i} \le n_{\ell-1} + \ydeg{g_{\ell-1}} - n_i = n_{\ell-1} + \ydeg{\lm{g_{\ell-1}}} - n_i,
\]
and the right-hand side is $n - n_i$.

\begin{lemma}
  \label{lem:block_sylvester_weak_popov}%
  The matrix \(M\) defined above has \(\deg(\det(M)) = \dd\) and is in
  \(s\)-weak Popov form for the shift \(s = (0, 1, \ldots, n-1)\).
\end{lemma}

\begin{proof}
  Let $0 \le i < \ell$ and $r = n_0 + \cdots + n_{i-1} = \ydeg{\lm{g_i}}$. Fix
  \(0 \le j < n-n_i\) and recall that $M_{rj}$ is the \(\xring\)-coefficient of
  \(y^j\) in \(g_i \in \xring[y]\). Thus all monomials of \(y^j M_{rj}\) appear
  in \(g_i\), and the leading monomial of \(g_i\) appears in \(y^r M_{rr}\).
  We deduce
  \[
    \deg(M_{rj}) + j = \tdeg{y^j M_{rj}} \le \tdeg{\lm{g_{i}}} = \deg(M_{rr}) + r,
  \]
  since \(\ledrl\) refines the total degree. Moreover, this inequality is
  strict when \(j > r\): otherwise \(g_i\) would contain a monomial \(y^j x^t\)
  with \(t = \tdeg{g_i} - j\), which is absurd since \(\lm{g_i} = y^r x^{t+j-r}
  \ltdrl y^j x^t\) when \(j > r\). Since for \(0 \le k < n_i\) the row
  \(M_{r+k,*}\) is equal to \(M_{r*}\) shifted to the right by \(k\) indices,
  we deduce that for all \(r,j\) in \(\{0,\ldots,n-1\}\), \(\deg(M_{rj}) + j
  \le \deg(M_{rr}) + r\) holds with strict inequality if \(j > r\). Thus
  the $s$-leading matrix of $M$ is lower triangular and invertible, meaning that
  $M$ is in $s$-weak Popov form.

  It follows from the above degree properties that the \(s\)-row degree of
  \(M_{r*}\) is \(\deg(M_{rr}) + s_r\) for \(0 \le r < n\). Hence,
  \cref{lem:degdet_degadj} ensures
  \(
    \deg(\det(M))
    = \sum_{0 \le r < n} \deg(M_{rr})
    = \sum_{0 \le i < \ell-1} n_i \xdeg{g_i}
  \),
  since \(\xdeg{g_{\ell-1}} = 0\). By construction of the \(n_i\)'s,
  the latter sum is the number of monomials that are not divisible by one
  of the \(\lm{g_i}\)'s: this is the ideal degree \(\dd =
  \dim_\field(\xyring/\ideal)\) (see \cref{sec:change_order:prelim}).
\end{proof}

\cref{lem:block_sylvester_basis} shows that \(M\) is a basis of
$\ideal_{<(\cdot,n)} = \xyring_{<(\cdot,n)} \cap \ideal$, which is a
\(\xring\)-submodule of \(\xyring_{<(\cdot,n)}\), itself seen as
\(\xmatspace{1}{n}\) via the basis \((1,y,\ldots,y^{n-1})\). This core property
makes the lexicographic basis of \(\ideal\) appear in the HNF of \(M\)
(\cref{cor:get_lexicographic}). Moreover, \(M\) has displacement rank \(\le
\ell\) with explicit displacement generators
(\cref{lem:change_order:displacement}), making the HNF computation efficient
via \cref{sec:structured_system,sec:change_order}. As highlighted in
\cref{thm:hnf}, thanks to the knowledge of \(\deg(\det(M))\), the computation
is faster if \(\hyprow\) and/or \(\hypcol\) holds; in
\cref{cor:get_lexicographic}, we prove that \(\hyprow\) coincides with the so-called
\emph{shape position} assumption on the lexicographic basis of \(\ideal\).

\begin{lemma}
  \label{lem:block_sylvester_basis}%
  The matrix \(M\) defined above is a basis of $\ideal_{<(\cdot,n)}$.
\end{lemma}

\begin{proof}
  Since \(\ideal\) is zero-dimensional, \(\ideal_{<(\cdot,n)}\) admits a module
  basis \(P\) which is a nonsingular matrix in \(\xmatspace{n}{n}\)
  \cite[Lem.\,3.4]{BerthomieuNeigerSafey2022}. Each row of $M$ represents a
  polynomial in $\ideal_{<(\cdot,n)}$, hence \(M = UP\) for some nonsingular
  \(U \in \xmatspace{n}{n}\). It remains to prove that \(U\) is unimodular; for
  this, it suffices to prove \(\deg(\det(P)) = \dd\), since the latter is
  \(\deg(\det(M))\) according to \cref{lem:block_sylvester_weak_popov}.
  Applying \cite[Lem.\,2.3]{NeigerVu2017}, we get \(\deg(\det(P)) =
  \dim_\field(\xyring_{<(\cdot,n)}/\ideal_{<(\cdot,n)})\). Now, recall from
  \cref{sec:change_order:prelim} that a vector space basis of
  \(\xyring/\ideal\) is formed by the \(\dd\) monomials that are not divisible
  by any of the \(\lm{g_i}\)'s. Since all these monomials are in
  \(\xyring_{<(\cdot,n)}\), they also form a vector space basis of
  \(\xyring_{<(\cdot,n)}/\ideal_{<(\cdot,n)}\), hence \(\deg(\det(P)) = \dd\).
\end{proof}

\subsection{Applying Algorithm~\ref{algo:HermiteFormSubmatrix}; proof of Theorem \ref{thm:change_order}}
\label{sec:change_order:contrib}

The next corollary of \cref{lem:block_sylvester_basis}, along with the
structure of \(M\) described in \cref{lem:change_order:displacement},
show that we can directly apply \algoName{algo:HermiteFormSubmatrix}
in order to compute \(\Glex\).

\begin{corollary}
  \label{cor:get_lexicographic}%
  Let \(m \in \{1,\ldots,n\}\) and let \(B\) be the \(m\times m\) leading
  principal submatrix of the HNF of \(M\). The elements of \(y\)-degree \(< m\)
  of the lexicographic Gr\"obner basis of \(\ideal\) can be read off from the
  rows of \(B\). Moreover, \(\ideal\) is in shape position if and only if
  \(M\) satisfies \(\hyprow\).
\end{corollary}

\begin{proof}[Sketch of proof]
  The last equivalence follows from the first claim. The first claim follows
  from \cite[Thm.\,4.1]{BerthomieuNeigerSafey2022}, since \(M\) is a basis of
  \(\ideal_{<(\cdot,n)}\) by \cref{lem:block_sylvester_basis} and since all
  monomials appearing in the lexicographic basis of \(\ideal\) have
  \(y\)-degree \(\le \max_i (\ydeg{g_i}) < n\).
\end{proof}

To make the above matrix \(M\) suit the Hankel-type displacement of
\cref{sec:structured_system}, we reverse the order of rows and consider \(LM\).
Since \(L\) is unimodular, the HNF of \(LM\) is that of \(M\), and
\(\deg(\det(LM)) = \dd\). In addition, from \(\adj{L} = (-1)^n L\) we get
\(\adj{LM} = (-1)^n\adj{M}L\) and thus \(\deg(\adj{LM}) = \deg(\adj{M}) \le \dd
+ n\) by \cref{lem:degdet_degadj}.

\begin{lemma}
  \label{lem:change_order:displacement}%
  The displacement rank of $LM$ is at most $\ell$. Using \(\bigO{\ell n \dd}\)
  operations in \(\field\), from \(\mathcal{G}\) one can compute displacement
  generators $\genl,\genr \in \xmatspace{\ell}{n}$ for $LM$, both of
  degree at most \(\dd\).
\end{lemma}

\begin{proof}
  We are interested in \(Z_0 LM - LM\trsp{Z_1}\). The rows of $Z_0 L M$
  represent the polynomials of \(\bar{\mathcal{G}}\), mirrored and shifted by
  one index, that is, \((0,y^{n_{\ell-1}-1}g_{\ell-1}, \ldots,yg_0)\). To
  describe the entries of $L M \trsp{Z_1}$, we define the operator
  $\psi:\xyring_{<(\cdot,n)} \to \xyring_{<(\cdot,n)}$ such that for $f =
  \sum_{i < n} f_i(x)y^{i}$, $\psi(f) = y(f - f_{n-1} y^{n-1}) + f_{n-1}$. Then
  the rows of $LM \trsp{Z_1}$ represent $(\psi(y^{n_{\ell-1}-1}g_{\ell-1}),
  \ldots, \psi(yg_0), \psi(g_0))$. As a result, the rows of \(Z_0 LM -
  LM\trsp{Z_1}\) represent the polynomials
  \begin{equation*}
    % \label{eq:appli:ZLMLMZ}%
    \begin{bmatrix}
      -\psi(y^{n_{\ell-1}-1}g_{\ell-1}) \\
      y^{n_{\ell-1}-1}g_{\ell-1} - \psi(y^{n_{\ell-1}-2}g_{\ell-1}) \\
      \vdots \\
      y^2 g_0 - \psi(yg_0) \\
      yg_0 - \psi(g_0)
    \end{bmatrix}
    \in \xyring_{<(\cdot,n)}^n.
  \end{equation*}
  There are two kinds of rows here. First, those of the form $y^{k+1}g_i -
  \psi(y^k g_i)$; since $\ydeg{y^{k}g_i} < n-1$, such rows are zero. Second,
  the \(\ell\) rows of the form $g_{i+1} - \psi(y^{n_i - 1} g_i)$ together with
  the very first row. These are the only rows that are possibly nonzero,
  implying that the displacement rank is at most \(\ell\). Their indices are
  $m_k = n - \sum_{i=0}^{k} n_i$ for $0 \le k < \ell$. Hence we can factor
  \(Z_0 LM - LM\trsp{Z_1} = \genl\trsp{\genr}\) with $\genl,\genr \in
  \xmatspace{n}{\ell}$, where the \(k\)th column of \(\genl\) is the \(m_k\)th
  column of the identity \(\ident{n}\) and the rows of \(\trsp{\genr}\) represent the
  \(\ell\) polynomials
  \[
    -\psi(y^{n_{\ell-1}-1}g_{\ell-1}),
    g_{\ell-1} - \psi(y^{n_{\ell-2}-1}g_{\ell-2}),
    \ldots,
    g_1 - \psi(y^{n_{0}-1}g_{0}).
  \]
  These polynomials have \(x\)-degree at most \(\max_{0\le i<\ell}
  \xdeg{g_i} \le \dd\), hence \(\deg(\genr) \le \dd\).
  The same \(x\)-degree bound also shows that
  the above \(\ell\) polynomials can be determined in time \(\bigO{\ell n \dd}\).
\end{proof}

From these, one obtains \cref{thm:change_order} as a consequence of
\cref{thm:hnf}. Indeed, using parameters \(\alpha = \ell\), \(n \leq 2d_y\),
\(\da = \dd + n\) and \(d_y \in \bigO{\dd}\), one gets the bound
\(\softO{\lceil m / \ell \rceil\ell^{\expmm-1}d_y \dd  + m^{\expmm-1}
d_y \dd }\); separating cases \(\ell \le m\) and \(\ell \ge m\)
shows that this is in \(\softO{(\ell^{\expmm-1}  + m^{\expmm-1}) d_y \dd}\).
Note that the parameter \(m\)
in \cref{thm:change_order}, does not need to
be known in advance. Indeed, since \(\ideal\) is zero-dimensional with known
degree \(\dd\), one can always check a posteriori whether \(m\) was chosen
large enough, by inspecting the leading monomials of the obtained polynomials.
Then, one may double \(m\) until \(\Glex\) is found; the impact on the
cost bound is at most a logarithmic factor.

%%%%%%%%%%%%%%%%%%%%%%%%%%%%%%%%%%%%%%%%%%%%%%%%%%%%%%%%%%%%%%%%%%%%%%%%%%%%%%%%%%%
%
%              BIBLIO
%
%%%%%%%%%%%%%%%%%%%%%%%%%%%%%%%%%%%%%%%%%%%%%%%%%%%%%%%%%%%%%%%%%%%%%%%%%%%%%%%%%%%
% Fakesection bibliography  (keep line, thanks)

\begin{acks}
  The authors are supported by
  the joint ANR-FWF ANR-22-CE91-00075 EAGLES project
  and the ANR-23-CE48-0003 CREAM project.
\end{acks}


\newcommand{\Gathen}{\relax}

%%%%%%%%%%%%%%%%%%%%%%%%%%%%%%%%%%%%%%%%%%%%%%%%%%%%%%%%%%%%%%%%%%%%%%%%%%%%%%%%%%%
%
%              APPENDIX
%
%%%%%%%%%%%%%%%%%%%%%%%%%%%%%%%%%%%%%%%%%%%%%%%%%%%%%%%%%%%%%%%%%%%%%%%%%%%%%%%%%%%

% Fakesection appendices  (keep line, thanks)
\appendix

\section{Matrix construction in Section~\ref{sec:change_order:buildM}: example}
\label{app:biv_matrix_example}

In this appendix, we illustrate the ideal properties (\cref{fig:gb_monomials})
and matrix construction (\cref{fig:matrix_degrees}) from
\cref{sec:change_order}.

\begin{figure}[ht]
  \centering
  \begin{tikzpicture}[scale=0.45,
    lmgb/.style={pattern={Dots[angle=45,distance=2pt]}, pattern color=black!20},
    lmshift/.style={pattern={Lines[yshift=.5pt,distance=2pt]}, pattern color=black!20},
    monomial/.style={pattern={Lines[angle=-45,yshift=-.5pt,distance=2.2pt]}, pattern color=black!20},
    ]

    % == GRID ==
    \draw[step=1,gray!35,very thin] (0,0) grid (14,13);

    % monomial basis
    % rows 0,1,2
    \foreach \x in {0,...,9} {
      \foreach \y in {0,1,2} {
        \fill[black!10,opacity=0.5] (\x,\y) rectangle ++(1,1);
      }
    }

    % row 3
    \foreach \x in {0,...,7} {
      \fill[black!10,opacity=0.5] (\x,3) rectangle ++(1,1);
    }

    % rows 4,5,6
    \foreach \x in {0,...,4} {
      \foreach \y in {4,5,6} {
        \fill[black!10,opacity=0.5] (\x,\y) rectangle ++(1,1);
      }
    }

    % == STAIRCASE ==
    \draw[line width=0.2mm] (0,7) -- (5,7) -- (5,4) -- (8,4) -- (8,3) -- (10,3) -- (10,0);

    \draw[dashed] (5,0) -- (5,7);
    \draw[dashed] (8,0) -- (8,4);
    \draw[dashed] (10,0) -- (10,3);

    \draw[<->] (0,-0.3) -- (5,-0.3);
    \node[below] at (2.5,-0.3) {\small $n_0=5$};
    \draw[<->] (5,-0.3) -- (8,-0.3);
    \node[below] at (6.5,-0.3) {\small $n_1=3$};
    \draw[<->] (8,-0.3) -- (10,-0.3);
    \node[below] at (9,-0.3) {\small $n_2=2$};
    \draw[<->] (10,-0.3) -- (12,-0.3);
    \node[below] at (11,-0.3) {\small $n_3=2$};

    % == MONOMIALS IN SHIFTS ==
    \fill[lmgb] (0,7) rectangle ++(1,1);
    \fill[monomial] (0,8) rectangle ++(1,1);
    \fill[monomial] (0,9) rectangle ++(1,1);
    \fill[monomial] (0,10) rectangle ++(1,1);
    \fill[monomial] (0,11) rectangle ++(1,1);

    \fill[lmshift] (1,7) rectangle ++(1,1);
    \fill[monomial] (1,8) rectangle ++(1,1);
    \fill[monomial] (1,9) rectangle ++(1,1);
    \fill[monomial] (1,10) rectangle ++(1,1);
    \fill[monomial] (1,11) rectangle ++(1,1);

    \fill[lmshift] (2,7) rectangle ++(1,1);
    \fill[monomial] (2,8) rectangle ++(1,1);
    \fill[monomial] (2,9) rectangle ++(1,1);
    \fill[monomial] (2,10) rectangle ++(1,1);

    \fill[lmshift] (3,7) rectangle ++(1,1);
    \fill[monomial] (3,8) rectangle ++(1,1);
    \fill[monomial] (3,9) rectangle ++(1,1);

    \fill[lmshift] (4,7) rectangle ++(1,1);
    \fill[monomial] (4,8) rectangle ++(1,1);

    \fill[lmgb] (5,4) rectangle ++(1,1);
    \fill[monomial] (5,5) rectangle ++(1,1);
    \fill[monomial] (5,6) rectangle ++(1,1);
    \fill[monomial] (5,7) rectangle ++(1,1);

    \fill[lmshift] (6,4) rectangle ++(1,1);
    \fill[monomial] (6,5) rectangle ++(1,1);
    \fill[monomial] (6,6) rectangle ++(1,1);

    \fill[lmshift] (7,4) rectangle ++(1,1);
    \fill[monomial] (7,5) rectangle ++(1,1);

    \fill[lmgb] (8,3) rectangle ++(1,1);
    \fill[monomial] (8,4) rectangle ++(1,1);

    \fill[lmshift] (9,3) rectangle ++(1,1);

    \fill[lmgb] (10,0) rectangle ++(1,1);
    \fill[monomial] (10,1) rectangle ++(1,1);

    \fill[lmshift] (11,0) rectangle ++(1,1);

    % == LEADING MONOMIALS ==

    % y^10, y^8 x^3, y^5 x^4, x^7
    \draw[->] (3.5,10.2) -- (0.5,7.5);
    \draw[->] (7.5,6.2) -- (5.5,4.5);
    \draw[->] (9.5,4.2) -- (8.5,3.5);
    \draw[->] (11.5,1.2) -- (10.5,0.5);
    \node[right] at (3.4,10.4) {\small \(x^7 = \lm{g_0}\)};
    \node[right] at (7.4,6.4)  {\small \(x^4 y^5 = \lm{g_1}\)};
    \node[right] at (9.4,4.4)  {\small \(x^3 y^8 = \lm{g_2}\)};
    \node[right] at (11.4,1.4) {\small \(y^{10} = \lm{g_3}\)};

    % == x-AXIS / y-AXIS ==
    \draw[->,thick] (-0.5,0) -- (14.8,0) node[right] {$y$};
    \draw[->,thick] (0,-0.9) -- (0,13.4) node[above] {$x$};
  \end{tikzpicture}
  \caption{\textmd{\emph{Representation of monomials in \(\bar{\mathcal{G}}\) for an ideal
        \(\ideal\) with a minimal \(\ledrl\)-Gr\"obner basis of \(\ell=4\) elements
        \((g_0,g_1,g_2,g_3)\) whose leading monomials are \((x^7, x^4 y^5,
        x^3 y^8, y^{10})\). The greyed area is the
        \(\ledrl\)-monomial basis of \(\xyring/\ideal\); the dotted monomials are
        the above \(\lm{g_i}\) for \(0 \le i < \ell\); the
        horizontally striped ones are the other leading monomials of
        \(\bar{\mathcal{G}}\), that is, \(\lm{y^k g_i}\) for \(1 \le k < n_i\)
        and \(0 \le i < \ell\); and the diagonally striped ones are all other
        monomials that possibly appear in \(\bar{\mathcal{G}}\).
        Here, \(n = n_0 + n_1 + n_2 + n_3 = 12\) and \(\dd = n_0 \xdeg{g_0} + n_1
        \xdeg{g_1} + n_2 \xdeg{g_2} = 5\cdot 7 + 3 \cdot 4 + 2 \cdot 3 = 53\).
  }}}
  \label{fig:gb_monomials}
  \bigskip
  \scalebox{0.84}{
    \(
    \bbordermatrix{
              & 1 & y & y^2 & y^3 & y^4 & y^5 & y^6 & y^7 & y^8 & y^9 & y^{10} & y^{11} \cr
   \hphantom{y^0}g_0 & \mycircled{7}  & 5  & 4 & 3 & 2 & 1 & 0 \cr
  \hphantom{^0}y g_0 &    & \myboxed{7}  & 5 & 4 & 3 & 2 & 1 & 0  \cr
      y^2 g_0        &    &    & \myboxed{7} & 5 & 4 & 3 & 2 & 1 & 0  \cr
      y^3 g_0        &    &    &   & \myboxed{7} & 5 & 4 & 3 & 2 & 1 & 0  \cr
      y^4 g_0        &    &    &   &   & \myboxed{7} & 5 & 4 & 3 & 2 & 1 & 0  \cr
  \hphantom{y^0} g_1 & 9  & 8  & 7 & 6 & 5 & \mycircled{4} & 2 & 1 & 0 \cr
  \hphantom{^0}y g_1 &    & 9  & 8  & 7 & 6 & 5 & \myboxed{4} & 2 & 1 & 0  \cr
       y^2 g_1       &    &    & 9  & 8  & 7 & 6 & 5 & \myboxed{4} & 2 & 1 & 0    \cr
   \hphantom{y^0}g_2 & 11 & 10 & 9 & 8 & 7 & 6 & 5 & 4 & \mycircled{3} & 1 & 0   \cr
  \hphantom{^0}y g_2 &    & 11 & 10 & 9 & 8 & 7 & 6 & 5 & 4 & \myboxed{3} & 1 & 0   \cr
  \hphantom{y^0}g_3  & 10 & 9  & 8 & 7 & 6 & 5 & 4 & 3 & 2 & 1 & \mycircled{0} \cr
  \hphantom{^0}y g_3 &    & 10 & 9  & 8 & 7 & 6 & 5 & 4 & 3 & 2 & 1 & \myboxed{0} \cr
    }
    \)
  }
  %% before adding circles/boxes:
  % \bbordermatrix{
  %               & 1 & y & y^2 & y^3 & y^4 & y^5 & y^6 & y^7 & y^8 & y^9 & y^{10} & y^{11} \cr
  %  \hphantom{y^0}g_0  & 7  & 5  & 4 & 3 & 2 & 1 & 0 \cr
  %  \hphantom{^0}y g_0 &    & 7  & 5 & 4 & 3 & 2 & 1 & 0  \cr
  %      y^2 g_0        &    &    & 7 & 5 & 4 & 3 & 2 & 1 & 0  \cr
  %      y^3 g_0        &    &    &   & 7 & 5 & 4 & 3 & 2 & 1 & 0  \cr
  %      y^4 g_0        &    &    &   &   & 7 & 5 & 4 & 3 & 2 & 1 & 0  \cr
  %  \hphantom{y^0} g_1 & 9  & 8  & 7 & 6 & 5 & 4 & 2 & 1 & 0 \cr
  %  \hphantom{^0}y g_1 &    & 9  & 8  & 7 & 6 & 5 & 4 & 2 & 1 & 0  \cr
  %       y^2 g_1       &    &    & 9  & 8  & 7 & 6 & 5 & 4 & 2 & 1 & 0    \cr
  %   \hphantom{y^0}g_2 & 11 & 10 & 9 & 8 & 7 & 6 & 5 & 4 & 3 & 1 & 0   \cr
  %  \hphantom{^0}y g_2 &    & 11 & 10 & 9 & 8 & 7 & 6 & 5 & 4 & 3 & 1 & 0   \cr
  %  \hphantom{y^0}g_3  & 10 & 9  & 8 & 7 & 6 & 5 & 4 & 3 & 2 & 1 & 0 \cr
  %  \hphantom{^0}y g_3 &    & 10 & 9  & 8 & 7 & 6 & 5 & 4 & 3 & 2 & 1 & 0 \cr
  % }
  \caption{\textmd{\emph{For an ideal \(\ideal\) as in
        \cref{fig:gb_monomials}, the matrix \(M\) constructed in
        \cref{sec:change_order:buildM} is in \(\xmatspace{12}{12}\).
        Its columns are indexed by \((1,y,\ldots,y^{11})\) and it has a
        block-Toeplitz structure, with rows indexed by the polynomials \(y^k g_i\)
        for \(1 \le k < n_i\) and \(0 \le i < \ell\). The row corresponding to
        \(y^k g_i\) stores the \(\xring\)-coefficients of this polynomial. The
        entries of \(M\) have degree less than or equal to the corresponding
        number in this illustrative ``degree matrix''. The circled numbers
        correspond to the leading monomials \(\lm{g_i}\), whereas the boxed
        numbers correspond to the other leading monomials of
        \(\bar{\mathcal{G}}\). For these diagonal entries, the numbers in the
        degree matrix are the exact degrees in \(M\), which implies that \(M\)
  is in \((0,1,\ldots,n-1)\)-weak Popov form.}}}
  \Description{Representation of monomials appearing in the extended Gr\"obner basis.
  Shape and degrees of the matrix M.}
  \label{fig:matrix_degrees}
\end{figure}

\section{Algorithms operating on the left}
\label{app:left_rows_versions}

The same result as \cref{prop:StructuredModularSystem} holds for left solving,
that is, computing \(X M^{-1} \bmod A\) for some \(X \in \xmatspace{m}{n}\) of
degree less than \(\Delta\). This only requires a minor modification of
\cref{algo:StructuredRightSolve-WithPoints}, at \cref{algo:StructuredRightSolve-WithPoints:mul},
where one then needs to perform the multiplication \(X(a_i) M(a_i)^{-1}\). For
this, one can rely on the algorithm of
\cite[Sec.\,5]{BostanJeannerodMouilleron2017} in order to compute
\(\itrsp{M(a_i)} \trsp{X(a_i)}\) efficiently, since \(\itrsp{M(a_i)}\) is
\((\trsp{Z_1},Z_0)\)-structured with generators \((-\bar{H},\bar{G})\). Indeed,
\begin{align*}
  \trsp{Z_0} M(a_i)^{-1} - M(a_i)^{-1} Z_1 &= \bar{G} \trsp{\bar{H}} \\
  {}\Leftrightarrow \trsp{Z_1} \itrsp{M(a_i)} - \itrsp{M(a_i)} Z_0 &= (-\bar{H}) \trsp{\bar{G}}.
\end{align*}
The corresponding algorithm has the specification in
\cref{algo:StructuredLeftSolve-WithPoints}.

\begin{algorithm}[ht]
	\algoCaptionLabel{StructuredLeftSolve-WithPoints}{G,H,X,(a_i)_{0 \le i < \Delta}}
  \begin{algorithmic}[1]

    \Require{%
        \algoitem displacement generators $G,H \in\xmatspace{n}{\alpha}$ that represent
        $M\in\xmatspace{n}{n}$ through \(Z_0 M - M \trsp{Z_1} = G \trsp{H}\) (see \cref{sec:structured_system:prelim}); \\
        \algoitem a matrix $X\in\xmatspace{m}{n}$ of degree less than \(\Delta\); \\
        \algoitem pairwise distinct points $a_0,\ldots,a_{\Delta-1} \in \field$, for \(\Delta \in \ZZp\).%
    }

    \Ensure{Any one of \FlagFail;   % bad luck in the structured linear algebra preconditioning
      \((i,\FlagSing)\) for some \(0 \le i < \Delta\);     % not Failure and \(\det(M)(a_i) = 0\) and \(\det(M)(a_j) \neq 0\) for \(j < i\)
      and $F \in \xmatspace{m}{n}_{<\Delta}$ such that $F M = X \bmod \prod_{0 \le i < \Delta} (x-a_i)$.}
  \end{algorithmic}
\end{algorithm}

Based on this variant,
one can directly adapt
\cref{algo:StructuredModularRightSolve} to left solving, yielding an
algorithm 
with the following specification and 
properties akin to those in
\cref{prop:StructuredModularSolve}.

\begin{algorithm}[ht]
  \algoCaptionLabel{StructuredModularLeftSolve}{G,H,X,\Delta}
  \begin{algorithmic}[1]

    \Require{%
      \algoitem displacement generators $G,H \in\xmatspace{n}{\alpha}$ that represent
      $M\in\xmatspace{n}{n}$ through \(Z_0 M - M \trsp{Z_1} = G \trsp{H}\) (see \cref{sec:structured_system:prelim}); \\
      \algoitem a matrix $X\in\xmatspace{m}{n}$ of degree less than \(\Delta\); \\
      \algoitem an integer \(\Delta \in \ZZp\) such that \(\field\) has cardinality \(\ge \Delta\).%
    }

    \Ensure{Any one of \FlagFail;
      % or pairwise distinct points $(a_i)_{0\le i < \Delta}$ in $\field$
      $((a_j)_{0\le j \le i}, \FlagSing)$
      with \(0 \le i < \Delta\) and $a_j \in \field$;
      %where \(0\le i < \Delta\)
      %and the $a_j$'s are pairwise distinct with \(\det(M)(a_i) = 0\)
      %and \(\det(M)(a_j) \neq 0\) for \(j < i\),
      and \((A,F)\) where \(A \in \xring_\Delta\)
      is coprime with \(\det(M)\) and $F \in \xmatspace{m}{n}_{<\Delta}$
      is such that $FM = X \bmod A$.}
  \end{algorithmic}
\end{algorithm}

Finally, using 
\cref{algo:StructuredModularLeftSolve},
a 
straightforward 
modification of
\cref{algo:InverseCols} gives an algorithm 
with the following specification and
with properties like those in \cref{prop:columns_of_inverse}.

\begin{algorithm}[ht]
	\algoCaptionLabel{InverseRows}{G,H,\indset,\dd,\da}
  \begin{algorithmic}[1]
    \Require{%
        \algoitem displacement generators $G,H \in\xmatspace{n}{\alpha}$ that represent
        $M\in\xmatspace{n}{n}$ through \(Z_0 M - M \trsp{Z_1} = G \trsp{H}\) (see \cref{sec:structured_system:prelim}); \\
        \algoitem \(\indset = (j_0,\ldots,j_{m-1})\) with \(0 \le j_0 < j_1 < \cdots < j_{m-1} < n\); \\
        \algoitem degree bounds \(\dd \ge \deg(\det(M))\) and \(\da \ge \deg(\adj{M})\).%
    }
    \Ensure{Any one of \FlagFail;
      \FlagSing;
      and \((\mu, \mu N_{\indset*})\) where \(N = M^{-1}\) and
    \(\mu\) is the monic least common denominator of \(N_{\indset*}\).}
  \end{algorithmic}
\end{algorithm}

\section{Efficient relation basis computations}
\label{app:relbas}

We first prove the next result, which implies \cref{prop:HNFRelBas}.

\begin{proposition}
  \label{prop:PopovRelBas}%
  There is an algorithm which, given \(\dd \in \ZZp\), $\mu \in\xring_{\dd}$,
  $F \in\xmatspace{m}{n}_{<\dd}$, and \(s \in \ZZ^m\), computes the \(s\)-Popov
  basis of $\mathcal{R}(\mu,F)$ using \(\softO{m^{\expmm-1}n\dd}\) operations
  in \(\field\).
\end{proposition}
\begin{proof}
  For \(n \in \bigO{m}\), this is a particular case of
  \cite[Thm.\,1.4]{Neiger2016}. Now consider \(m \le n\). Using
  \(\softO{m^{\expmm-1} n \dd}\) operations in \(\field\), one can compute a
  matrix \(C \in \xmatspace{m}{r}\) of degree \(\le \deg(F)\) whose columns
  form a basis of \(\cspan{F}\) \cite{ZhouLabahn2013}, where \(r = \rank{F}\).
  It is enough to prove $\mathcal{R}(\mu,F) = \mathcal{R}(\mu,C)$,
  since \(C\) has \(r \in \bigO{m}\) columns and thus computing the
  \(s\)-Popov basis of $\mathcal{R}(\mu,C)$ costs \(\softO{m^{\expmm-1} r
  \dd}\) according to \cite[Thm.\,1.4]{Neiger2016}.
  To prove this equality of relation modules, note that
  by definition of \(C\),
  there exists a unimodular \(V \in \xmatspace{n}{n}\)
  %% NOTE if ``there exists'' is too quick, here are more arguments:
  %% by def of C, one can write
  %%      F = C * U1, for some U1  r x n;
  %%      F * U2 = C, for some U2  n x r;
  %% then one has C = C * U1 * U2, and since C has full column rank
  %% this shows that U1 * U2 = I_r.
  %% In particular U1 has unimodular column bases (i.e., it is saturated, etc.)
  %% so it can be ``row-completed'' into an
  %% n x n unimodular matrix V = [U1 \\ V2]  (cf Zhou-Labahn 2014, for example)
  %% We have [C \;\; 0] V = F, so F * V^{-1} = [C \;\; 0].
  such that \(F V = [C \;\; 0]\).
  %We write \(V = [V_0 \;\; V_1]\) with \(V_0\)
  %having \(r\) columns, so that \(F V_1 = C\).
  Then, for \(p \in \xmatspace{1}{m}\), we have
  \begin{align*}
    pF = 0 \bmod \mu & \Leftrightarrow pFV = \mu q V \text{ for some } q \in \xmatspace{1}{n} \\
                     & \Leftrightarrow [pC \;\; 0] = \mu q \text{ for some } q \in \xmatspace{1}{n} \\
                     & \Leftrightarrow pC = \mu q \text{ for some } q \in \xmatspace{1}{r},
  \end{align*}
  where the second equivalence exploits the unimodularity of \(V\).
\end{proof}

One particularly favorable case is when we seek the HNF basis and this basis is
known to satisfy \(\hyprow\). When \(n = 1\), the basis can be found through
explicit formulas for its first column (\cref{lem:hrow_relbas_onecol}), and for
\(n > 1\) it can be found through a sequence of \(n-1\) relation bases of a
small \(2 \times n\) matrix (\cref{lem:hrow_relbas_manycol}). These are the
results used at
\crefrange{algo:HermiteFormSubmatrix:if_hypcol_hyprow}{algo:HermiteFormSubmatrix:hrow_hcol_relbas}
and
\crefrange{algo:HermiteFormSubmatrix:gcd}{algo:HermiteFormSubmatrix:hrow_relbas},
respectively, of \cref{algo:HermiteFormSubmatrix}.

\begin{lemma}
  \label{lem:hrow_relbas_onecol}%
  Let $\mu \in\xring_{\dd}$ be monic and $c \in\xmatspace{m}{1}_{<\dd}$ for
  some \(\dd \in \ZZp\). If \(\gcd(c_{0},\mu) = 1\), then the HNF basis of
  $\mathcal{R}(\mu,c)$ is
  \[ B=
    \begin{bmatrix}
      \mu               \\
      b_{1} & 1  \\
      \vdots &  & \ddots \\
      b_{m-1} & & & 1
    \end{bmatrix}
    \in \xmatspace{m}{m},
  \]
  where \(b_{i} = -c_i/c_0 \bmod \mu\) for \(1 \le i < m\).
\end{lemma}

\begin{proof}
  Let \(H\) be the HNF basis of \(\mathcal{R}(\mu,c)\). Its first row
  has the form \([b_0 \; 0 \; \cdots \; 0]\) for a nonzero \(b_0 \in
  \xring\) such that \(\mu\) divides \(b_0 c_0\). As \(\gcd(\mu,c_0) = 1\),
  \(\mu\) divides \(b_0\). On the other hand, \([\mu \; 0 \; \cdots \;
  0]\) is in the row span of \(H\) which is lower triangular with nonzero
  diagonal entries, hence this vector is a multiple of \([b_0 \; 0
  \; \cdots \; 0]\). Thus \(\mu = b_0\).

  By construction, the row span of \(B\) is in \(\mathcal{R}(\mu,c)\).
  Conversely, let \(p \in \mathcal{R}(\mu,c)\).
  From \(p_0 c_0 + \cdots + p_{m-1}c_{m-1} = 0 \bmod \mu\) we deduce \(p_0 =
  p_1b_1 + \cdots + p_{m-1}b_{m-1} \bmod \mu\) since \(\gcd(c_0,\mu)=1\).
  This shows that \(p = [u \;\; p_1 \;\; \cdots \;\; p_{m-1}]B\) for some \(u \in
  \xring\), using the fact that \(b_0 = \mu\). We have proved \(p \in \rspan{B}\).
\end{proof}

\begin{lemma}
  \label{lem:hrow_relbas_manycol}%
  Let \(\dd \in \ZZp\),
  let $\mu \in\xring_{\dd}$ be monic, and let
  $R \in\xmatspace{m}{n}_{<\dd}$ be such that \(\gcd(\mu,R_{00},\ldots,R_{0,n-1})=1\).
  Let \(B \in \xmatspace{m}{m}\) be the HNF basis of \(\mathcal{R}(\mu,R)\).
  If \(B\) satisfies \(\hyprow\), then for any \(1 \leq i < m\), the HNF basis of
  \(\mathcal{R}(\mu,R_{(0,i),*})\) is
  \[
    B_{(0,i),(0,i)} =
    \begin{bmatrix}
      \mu & 0 \\ B_{i0} & 1
    \end{bmatrix}.
  \]
\end{lemma}

\begin{proof}
  Let \(0 \leq i < m\). First, note that \(B_{(0,i),(0,i)}\) is in HNF. 
  As
  \(B\) satisfies \(\hyprow\), we have \(B_{i0} R_{0*} + R_{i*} = 0 \bmod
  \mu\), hence both rows of \(B_{(0,i),(0,i)}\) are in
  \(\mathcal{R}(\mu,R_{(0,i),*})\). 
  Thus, 
  it remains to prove that any \(p \in
  \mathcal{R}(\mu,R_{(0,i),*})\) is also in \(\rspan{B_{(0,i),(0,i)}}\).
  Indeed, from \(p \in \mathcal{R}(\mu,R_{(0,i),*})\) we get
  \(p_0 R_{0*} + p_1 R_{i*} = 0 \bmod \mu\), and 
  therefore
  \((p_0 - p_1 B_{i0}) R_{0*} = 0 \bmod \mu\) by an above remark.
  Since \(\gcd(\mu, R_{0*}) = 1\), 
  this implies 
  \(p_0 = p_1 B_{i0} \bmod
  \mu\), 
  and the latter means 
  \(p \in \rspan{B_{(0,i),(0,i)}}\).
\end{proof}

\end{document}